\documentclass[11pt]{myclass}
\usepackage{amsmath,amsthm,amssymb}
\usepackage{paralist}
\usepackage{color}
\usepackage{verbatim} 

\usepackage{ifpdf}
\ifpdf    
\usepackage{hyperref}
\else    
\usepackage[hypertex]{hyperref}
\fi

\DeclareMathOperator*{\argmax}{arg\,max}

\newcommand{\Z}{{\mathbb{Z}}}
\newcommand{\R}{\reals}

\newcommand{\eps}{\epsilon}

\DeclareMathOperator{\dist}{dist}
\newcommand{\aset}[1]{\{ #1 \}}
\newcommand{\aanote}[1]{{\color{red}$\ll$\textsf{#1 --Alex}$\gg$\marginpar{\tiny\bf AA}}}
\newcommand{\rnote}[1]{}

\providecommand{\abs}[1]{\lvert#1\rvert}
\providecommand{\card}[1]{\lvert#1\rvert}
\providecommand{\norm}[1]{\lVert#1\rVert}

\DeclareMathOperator{\spn}{span}
\DeclareMathOperator{\var}{Var}
\newcommand{\Fquery}{F_{\mathrm{query}}}
\newcommand{\Fspace}{F_{\mathrm{space}}}
\newcommand{\Fprep}{F_{\mathrm{prep}}}
\newcommand{\whp}{w.h.p.\xspace}
\newcommand{\transpose}{{\rm T}}
\newcommand{\tran}{\transpose}

\newcommand{\tX}{{\tilde X}}
\newcommand{\tP}{{\tilde P}}
\newcommand{\tU}{{\tilde U}}
\newcommand{\tV}{{\tilde V}}
\newcommand{\mU}{{\mathcal U}}

\def\compactify{\itemsep=0pt \topsep=0pt \partopsep=0pt \parsep=0pt}

\newcommand{\compress}[1]{} 

\title{Spectral Approaches to Nearest Neighbor Search}
\author{
Amirali Abdullah
\thanks{%
A large of this work was carried out and funded at Microsoft Research. This effort was also supported in part by the NSF 
under grant \#CCF-0953066
Email: \texttt{amirali@cs.utah.edu} }
\\
University of Utah
\and
Alexandr Andoni\\
Microsoft Research
\and
Ravindran Kannan\\
Microsoft Research
\and
Robert Krauthgamer\thanks{%
Work done in part while at Microsoft Silicon Valley. 
Work supported in part by the Israel Science Foundation grant \#897/13, 
the US-Israel BSF grant \#2010418, and by the Citi Foundation.
Email: \texttt{robert.krauthgamer@weizmann.ac.il} }
\\
Weizmann Institute
}
\begin{document}

\maketitle


\begin{abstract}
We study spectral algorithms for the high-dimensional Nearest Neighbor
Search problem (NNS).  In particular, we consider a semi-random
setting where a dataset $P$ in $\R^d$ is chosen arbitrarily from an
unknown subspace of low dimension $k\ll d$, and then perturbed by 
fully $d$-dimensional Gaussian noise.  We design spectral NNS
algorithms whose query time depends polynomially on $d$ and $\log n$
(where $n=|P|$) for large ranges of $k$, $d$ and $n$.  Our algorithms
use a repeated computation of the top PCA vector/subspace, and are
effective even when the random-noise magnitude is {\em much larger}
than the interpoint distances in $P$.  Our motivation is that in
practice, a number of spectral NNS algorithms outperform the
random-projection methods that seem otherwise theoretically optimal on worst case datasets. In this paper we aim to provide theoretical justification for this disparity.

\end{abstract}

\newpage
\section{Introduction}

A fundamental tool in high-dimensional computational geometry is the
random projection method.  Most notably, the Johnson-Lindenstrass
Lemma \cite{JL} says that projecting onto a uniformly random
$k$-dimensional subspace of $\R^d$ approximately preserves the
distance between any (fixed) points $x,y\in \R^d$ (up to scaling),
except with probability exponentially small in $k$.  This turns out to
be a very powerful approach as it effectively reduces the dimension
from $d$ to a usually much smaller $k$ via a computationally cheap
projection, and as such has had a tremendous impact on algorithmic
questions in high-dimensional geometry.

A classical application of random projections is to the
high-dimensional Nearest Neighbor Search (NNS) problem. Here we are
given a dataset of $n$ points from $\R^d$, which we preprocess to
subsequently find, given a query point $q\in \R^d$, its closest point
from the dataset.  It is now well-known that exact NNS admits
algorithms whose running times have good dependence on $n$, but
exponential dependence on the dimension $d$ \cite{Mei93, Cl1}; however
these are unsatisfactory for moderately large $d$.

To deal with this ``curse of dimensionality'',
researchers have studied algorithms for \emph{approximate} NNS,
and indeed in the high-dimensional regime, many, if not all,
of these algorithms rely heavily on the random projection method.
Consider the case of Locality-Sensitive
Hashing (LSH), introduced in \cite{IM}, which has been a theoretically
and practically successful approach to NNS. All known variants of LSH
for the Euclidean space, including \cite{IM, DIIM, AI, AINR-subLSH},
involve random projections.%
\footnote{While \cite{IM} is designed for the Hamming space, their
  algorithm is extended to the Euclidean space by an embedding of
  $\ell_2$ into $\ell_1$, which itself uses random projections
  \cite{JS}.  } For example, the space partitioning algorithm of
\cite{DIIM} can be viewed as follows. Project the dataset onto a
random $k$-dimensional subspace, and impose a randomly-shifted uniform
grid. Then, to locate the near(est) neighbor of a point $q$, look up
the points in the grid cell where $q$ falls into. Usually, this space
partitioning is repeated a few times to amplify the probability of
success (see also \cite{Pan}).

While random projections work well and have provable guarantees, it is
natural to ask whether one can improve the performance by
replacing ``random'' projections with ``best'' projections.
Can one optimize the projection to use --- and the space partitioning more generally --- as a {\em function of the dataset at hand}?
For example, in some tasks requiring
dimension reduction, practitioners often rely on Principal
Component Analysis (PCA) and its variants.
Indeed, in practice, this consideration led to numerous successful
heuristics such as PCA tree \cite{Sproull-pcaTree, mcnames2001fast,
  verma2009spatial} and its variants (called randomized kd-tree)
\cite{silpa2008optimised, muja2009fast}, spectral hashing
\cite{weiss2008spectral}, semantic hashing
\cite{salakhutdinov2009semantic}, and WTA hashing
\cite{yagnik2011power}, to name just a few. Oftentimes, these
heuristics outperform algorithms based on vanilla random
projections. All of them adapt to the dataset, including many that
perform some {\em spectral} decomposition of the dataset. However, in
contrast to the random projection method, none of these methods have
rigorous correctness or performance guarantees.

Bridging the gap between random projections and data-aware projections
has been recognized as a big open question in Massive Data Analysis,
see e.g.\ a recent National Research Council report
\cite[Section~5]{NRC-report}.  The challenge here is that random
projections are themselves (theoretically) optimal not only for
dimensionality reduction \cite{Alon, jw11-JL}, but also for some of
its algorithmic applications \cite{IW03, W-optimalSpace}, including
NNS in certain settings \cite{AIP}. We are aware of only one work
addressing this question: data-dependent LSH, which was introduced
recently \cite{AINR-subLSH}, provably improves the query time
polynomially.  However, their \emph{space partitions} are very
different from the aforementioned practical heuristics (e.g., they are
not spectral-based), and do not explain why data-aware {\em
  projections} help at all.

In this work, we address this gap by studying data-aware projections for the nearest neighbor search problem. As argued above, for worst-case inputs we are unlikely to beat the
performance of random projections, and thus it seems justified to
revert to the framework of smoothed analysis \cite{ST09} to study the gap in practical performance.  We 
consider a semi-random model, where the dataset is formed by first
taking an arbitrary (worst-case) set of $n$ points in a
$k$-dimensional subspace of $\R^d$, and then perturbing each point by
adding to it Gaussian noise $N_d(0,\sigma^2 I_d)$.  The query point is
selected using a similar process.  Our algorithms are able to find the
query's nearest neighbor as long as there is a small gap ($1$ vs
$1+\eps$) in the distance to the nearest neighbor versus other points
in the {\em unperturbed} space --- this is a much weaker assumption
than assuming the same for the perturbed points. 

Most importantly, our results hold even when the noise magnitude is
{\em much larger} than the distance to the nearest neighbor.  
The noise vector has length (about) $\sigma\sqrt d$, 
and so for $\sigma \gg 1/\sqrt d$, the noise magnitude exceeds the original
distances. In such a case, a random Johnson-Lindenstrauss projection to
a smaller dimension will not work --- the error due to the projection will
lose all the information on the nearest neighbor.


In contrast, our results show that data-aware projections provably guarantee good performance in this model.
We describe the precise model in Section~\ref{sec:model}.


\subsection{Algorithmic Results}

We propose two spectral algorithms for nearest neighbor search, which
achieve essentially the same performance as NNS algorithms in $k$ and
$O(k\log k)$-dimensional space, respectively.  These spectral
algorithms rely on computing a PCA subspace or vector respectively ---
the span of the singular vectors corresponding to the top singular
values of an $n\times d$ matrix representing some $n$ points in
$\R^d$. Our algorithms are inspired by PCA-based methods that are
commonly used in practice for high-dimensional NNS, and we believe
that our rigorous analysis may help explain (or direct) those
empirical successes.
We defer the precise statements to the respective technical sections
(specifically, Theorems~\ref{thm:mainfour} and \ref{thm:pcaTree}),
focusing here on a {\em qualitative} description.

The first algorithm performs {\em iterative PCA}.  Namely, it employs
PCA to extract a subspace of dimension (at most) $k$, identifies the
points captured well by this subspace, and then repeats iteratively on
the remaining points.  
The algorithm performs at most $O(\sqrt{d\log n})$ PCAs in total, 
and effectively reduces the original NNS problem to $O(\sqrt{d\log n})$ 
instances of NNS in $k$ dimensions.  Each of these NNS instances
can be solved by any standard low-dimensional $(1+\eps)$-approximate
NNS, such as \cite{Cl2, AMNSW, arya2009space, HP, AMal, CG06}, which
can give, say, query time $(1/\eps)^{O(k)}\log^2n$. See
Section~\ref{sec:largeGaussian}, and the crucial technical tool it uses
in Section \ref{sec:sin-theta}. As a warmup, we initially introduce a
simplified version of the algorithm for a (much) simpler model in
Section~\ref{sec:bounded}.

The second algorithm is a variant of the aforementioned PCA tree,
and constructs a tree that represents a recursive space partition.
Each tree node corresponds to finding the top PCA direction,
and partitioning the dataset into slabs perpendicular to this direction.
We recurse on each slab until the tree reaches depth $2k$.
The query algorithm
searches the tree by following a small number of children (slabs) at
each node. This algorithm also requires an additional preprocessing
step that ensures that the dataset is not overly ``clumped''.  The overall
query time is $(k/\eps)^{O(k)}\cdot d^2$. 
See Section \ref{sec:pcaTree}.

While the first algorithm is randomized, the second algorithm is
deterministic and its failure probability comes only from the
semi-random model (randomness in the instance).

\subsection{Related Work}

There has been work on understanding how various tree data structures
adapt to a \emph{low-dimensional} pointset, including
\cite{dasgupta2008random, verma2009spatial}.
For example, \cite{verma2009spatial} show that PCA trees adapt to a form of
``local covariance dimension'', a spectrally-inspired notion of
dimension, in the sense that a PCA tree halves the ``diameter'' of the
pointset after a number of levels dependent on this dimension
notion (as opposed to the ambient dimension $d$).  Our work differs in
a few respects.  First, our datasets do not have a small local
covariance dimension.  Second, our algorithms have guarantees of
performance and correctness for NNS for a worst-case query point
(e.g., the true nearest neighbor can be any dataset point).  In
contrast, \cite{verma2009spatial} prove a guarantee on diameter
progress, which does not necessarily imply performance guarantees for
NNS, and, in fact, may only hold for average query point (e.g., when
the true nearest neighbor is random).  Indeed, for algorithms
considered in \cite{verma2009spatial}, it is easy to exhibit
cases where NNS fails.%
\footnote{For example, if we consider the top PCA direction of a
  dataset and the median threshold, we can plant a query--near-neighbor
  pair on the two sides of the partition. Then, this pair, which won't
  affect top PCA direction much, will be separated in the PCA tree
  right from the beginning.
}

For our model, it is tempting to design NNS algorithms that find the
original $k$-dimensional subspace and thus ``de-noise'' the dataset by
projecting the data onto it. This approach would require us to solve the
$\ell_\infty$-regression problem with high precision.%
\footnote{As we explain later, related problems,
such as $\ell_2$-regression would not be sufficient.
}
Unfortunately, this problem is NP-hard in general \cite{nphardcompute},
and the known algorithms are quite expensive, unless $k$ is constant.
Har-Peled and Varadarajan  \cite{harhigh}
present an algorithm for $\ell_\infty$-regression achieving
$(1+\eps)$-approximation in time $O(nd e^{e^{O(k^2)}\eps^{-2k-3}})$,
which may be prohibitive when $k \geq \Omega(\sqrt{\log \log n})$.
In fact, there is a constant $\delta>0$ such that it is Quasi-NP-hard,
i.e., implies $\mathrm{NP} \subseteq \mathrm{DTIME}(2^{(\log n)^{O(1)}})$,
to even find a $(\log n)^{\delta}$ approximation to the best fitting
$k$-subspace when $k \geq d^{\epsilon}$ for any fixed $\epsilon>0$
\cite{lowerboundfitting}.

We also note that the problem of finding the underlying
$k$-dimensional space is somewhat reminiscent of the learning mixtures
of Gaussians problem \cite{Dasgupta99}; see also \cite{AK,
  vempala2004spectral, moitra2010settling} and references therein.  In
the latter problem, the input is $n$ samples generated from a mixture
of $k$ Gaussian distributions with unknown mean (called centers),
and the goal is to identify these centers (the means of the $k$ Gaussians).
Our setting can be viewed as having $n$ centers in a $k$-dimensional subspace
of $\R^d$, and the input contains exactly one sample from each of the $n$
centers.
Methods known from learning mixtures of Gaussians rely crucially on
obtaining multiple samples from the same center (Gaussian),
and thus do not seem applicable here.

Finally, the problem of nearest neighbor search in Euclidean settings
with ``effective low-dimension'' has received a lot of attention in the
literature, including \cite{KRu,KLee,HM05,BKL06,CG06,IN07} among many
others. Particularly related is also the work \cite{HPK13-lowd}, where
the authors consider the case when the dataset is high-dimensional,
but the query comes from a (predetermined) low-dimensional
subspace. These results do not seem readily applicable in our
setting because our dataset is really high-dimensional,
say in the sense that the doubling dimension is $\Omega(d)$.

\subsection{Techniques and Ideas}

We now discuss the main technical ideas behind our two algorithms. First,
we explain why some natural ideas do {\em not} work. The very first
intuitive line of attack to our problem is to compute a
$k$-dimensional PCA of the pointset, project it into this
$k$-dimensional subspace, and solve the $k$-dimensional NNS problem
there. This approach fails because the noise is too large, and PCA
only optimizes {\em the sum of distances} (i.e., an average quantity,
as opposed to the ``worst case'' quantity). In particular, suppose
most of the points lie along some direction $\vec{u}$ and only a few
points lie in the remaining dimensions of our original subspace $U$
(which we call sparse directions). Then, the $k$-PCA of the dataset
will return a top singular vector close to $\vec{u}$, but the
remaining singular vectors will be mostly determined by random
noise. In particular, the points with high component along sparse
directions may be very far away from the subspace returned by our PCA,
and hence ``poorly-captured'' by the PCA space. Then, the NNS data
structure on the projected points will fail for some query
points. If the difference between a query point $q$ and its
nearest neighbor $p^*$ is along $\vec{u}$, whereas the difference
between $q$ and a poorly-captured point $p'$ is along the sparse
directions, then any such $p'$ will cluster around $q$ in the
$k$-PCA space, at a distance much closer than $\|q-p^*\|$. See Figure \ref{fig:poorC}.

\begin{figure}[H]
  \begin{center}
    \includegraphics[scale = 0.5]{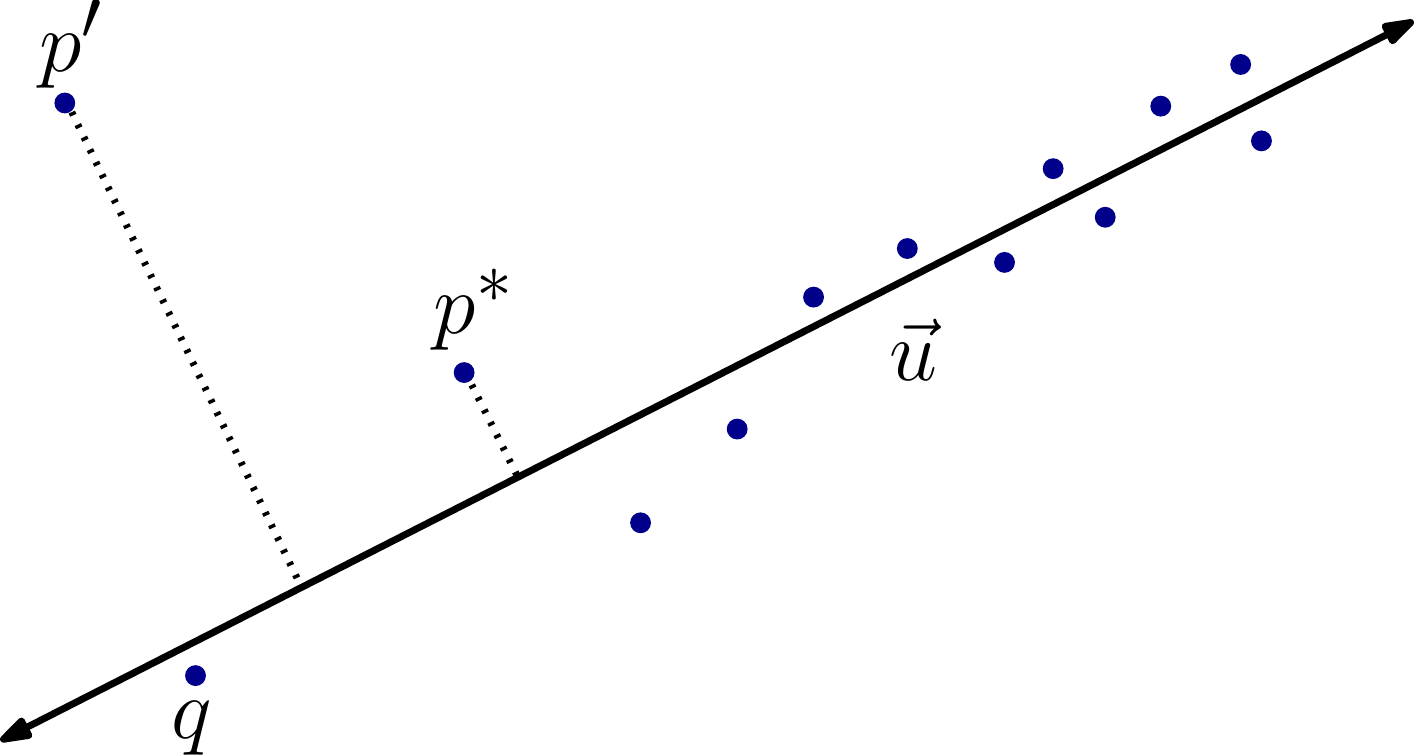}
  \end{center}
  \caption{$p'$ lies along a poorly captured direction.}
  \label{fig:poorC}
\end{figure}

Our first algorithm instead runs $k$-PCAs {\em iteratively}, while
pruning away points ``well-captured'' by the PCA (i.e., close to the
PCA subspace). In particular, this allows us to discover the points in
sparse directions in {\em later} iterations.  Furthermore, to ensure
correctness of nearest neighbor queries in presense of large noise, we
do not necessarily take all the top $k$ singular values, but only
those that exceed some threshold value; this guarantees that all
directions in our PCA subspace have a large component inside
$U$. Showing this guarantee analytically is non-trivial, starting even with 
the definition of what it means to be ``close'' for two spaces, which
may have different dimensions.  For this purpose, we employ the
so-called $\sin\theta$ machinery, which was developed by Davis and
Kahan \cite{DK70} and by Wedin \cite{wedin}, to bound the
perturbations of the singular \emph{vectors} of a matrix in presence
of noise.  Notice the difference from the more usual theory of
perturbations of the singular \emph{values}.  For example, in
contrast to singular values, it is not true that the top singular
vector is ``stable'' when we perturb a given matrix.

The actual algorithm has one more important aspect: in each iteration,
the PCA space is computed on a {\em sample} of the (surviving) data
points. This modification allows us to control spurious conditioning
induced by earlier iterations.  In particular, if instead we compute
the PCA of the full data, once we argue that a vector $\tilde p$
``behaves nicely'' in one iteration, we might effectively condition on
the direction of its noise, potentially jeopardizing noise concentration bounds on later
iterations. (While we do not know if sampling is really necessary for
the algorithm to work, we note that practically it is a very
reasonable idea to speed up preprocessing nonetheless.)

The second algorithm is based on the PCA-tree, which partitions the
space recursively, according to the top PCA direction. This can be
seen as another (extreme) form of ``iterative PCA''. At each node, the
algorithm extracts one top PCA direction, which always contains the
``maximum'' information about the dataset. Then it partitions the dataset
into a few slabs along this direction, thereby partitioning the datasets into smaller parts
``as quickly as possible''. This allows the tree to narrow down on the
sparse directions quicker. The performance of the PCA tree depends
exponentially on its depth, hence the crux of the argument is to bound
the depth. While it seems plausibly easy to show that a partitioning
direction should never be repeated, this would give too loose a bound,
as there could be a total of $\approx \exp(k)$ essentially distinct
directions in a $k$-dimensional space. Instead, we perform a mild form
of orthonormalization as we progress down the tree, to ensure only
$O(k)$ directions are used in total. In the end, the query time 
is roughly $k^{O(k)}$, i.e., equivalent to a NNS in an $O(k\log
k)$-dimensional space.

We note that this algorithm has two interesting aspects. First, one
has to use {\em centered} PCA, i.e., PCA on the data centered at zero:
otherwise, every small error in PCA direction may move points a lot
for subsequent iterations, misleading a non-centered PCA. Second, from
time to time, we need to do ``de-clumping'' of the data, which
essentially means that the data is sparsified if the points are too
close to each other. This operation also appears necessary; otherwise,
a cluster of points that are close in the original space, might
mislead the PCA due to their noise components. Furthermore, in contrast to
the first algorithm, we cannot afford to iterate through $\approx d$
iterations to eliminate ``bad'' directions one by one.



\section{The Model}
\label{sec:model}

We assume throughout the dataset is generated as follows.%
\footnote{An exception is the warm-up Section~\ref{sec:bounded}, where the noise is small adversarial.}
Let $U$ be a $k$-dimensional subspace of $\R^d$. 
Let $P=\aset{p_1,\ldots,p_n}$ be a set of $n$ points 
all living in $U$ and having at least unit norm,
and let $q \in U$ be a query point. 
We assume that $d=\Omega(\log n)$. The point set $P$ is perturbed to
create $\tilde P=\aset{\tilde p_1,\ldots,\tilde p_n}$ by adding to
each point independent Gaussian noise, and the query point $q$ is
perturbed similarly.  Formally,
\begin{align} 
  \label{eq:tildeP}
  \tilde p_i & = p_i + t_i 
  \, \text{ where } t_i \sim N_d(0,\sigma I_d),
  \qquad \forall p_i\in P, 
  \\
  \label{eq:tildeq}
  \tilde q & = q + t_q
  \; \text{ where } t_q \sim N_d(0,\sigma I_d). 
\end{align}

Let us denote the nearest neighbor to $q$ in $P$ by $p^*$ and let $\tilde p^*$ be its perturbed version.
We shall actually consider the {\em near-neighbor} problem,
by assuming that in the unperturbed space, there is one point $p^*\in P$ 
within distance $1$ from the query, and all other points are at distance at least $1+\eps$ from the query, for some known $0<\eps<1$. Formally,
\begin{equation} \label{eq:pstar}
  \exists p^*\in P \text{ such that }
  \norm{q-p^*}\leq 1 \text { and }
  \forall p\in P\setminus\{p^*\},\ 
  \norm{q-p} \geq 1+\eps. 
\end{equation}

We note that even if there is more than one such point $p^*$ so that $\norm{q - p^*} \leq 1$, our algorithms will return one of these close $p^*$ correctly. Also our analysis in Section \ref{sec:largeGaussian} extends trivially to show that for any $x$ such that $x \geq 1$ and $\|q-p^*\| = x$, our first algorithm the iterative PCA  actually returns a $(1+\eps)$-approximate nearest neighbor to $q$. We omit the details of this extended case for ease of exposition.

\subsection{Preliminary Observations}

For the problem to be interesting, we need that the perturbation does
not change the nearest neighbor, i.e., $\tilde p^*$ remains the
closest point to $\tilde q$.  We indeed show this is the case as long
as $\sigma\ll \eps / \sqrt[4]{d\log n}$.  Notice that the total noise
magnitude is roughly $\sigma\sqrt{d}$, which can be much larger than
$1$ (the original distance to the nearest neighbor).  Hence after the
noise is added, the ratio of the distance to the nearest neighbor and
to other (nearby) points becomes very close to $1$. This is the main
difficulty of the problem, as, for example, it is the case where
random dimensionality reduction would lose nearest neighbor
information. 
We recommend to keep in mind the following parameter settings:
$k=20$ and $\eps=0.1$ are constants,
while $d=\log^3n$ and $\sigma=\Theta(1/\log n)$ depend 
asymptotically on $n=\card{P}$. 
In this case, for example, our algorithms actually withstand noise
of magnitude $\Theta(\sqrt{\log n})\gg 1$.

Here and in the rest of the paper, we will repeatedly employ
concentration bounds for Gaussian noise, expressed as tail
inequalities on $\chi^2$ distribution.  We state here for reference
bounds from \cite{LM98}, where $\chi^2_d$ is the same distribution as
$\norm{N_d(0,I_d)}_2^2$.  The term \emph{with high probability (\whp)}
will mean that probability $1-n^{-C}$ for sufficiently large $C>0$.

\begin{theorem}(\cite{LM98})\label{thm:XiSquare}  
Let $X \sim \chi^2_{d}$. For all $x \geq 0$,
$$\Pr \left[X \geq d \left( 1+2 \sqrt{\tfrac{x}{d}} \right) + x \right] \leq e^{-x},
 \quad\mbox{ and }\quad
 \Pr \left[X \leq d \left(1 -2\sqrt{\tfrac{x}{d}} \right) \right] \leq e^{-x}.$$
\end{theorem}
\begin{corollary}\label{cor:projec}
For $n\ge 1$,  let $X \sim \chi^2_{d}$.
Then $\Pr[|X-d| \geq 
d+  4 \sqrt{d \log n} + 4 \log n] \leq \tfrac{2}{n^4}$.
\end{corollary}

We now show that after the perturbation of $P,q$, the nearest neighbor
of $\tilde q$ will remain $\tilde p^*$, \whp

\begin{lemma} \label{lem:NNSstays}
Consider the above model \eqref{eq:tildeP}-\eqref{eq:pstar}
for $n>1$, $\eps\in(0,1)$, dimensions $k<d = \Omega(\log n)$, 
and noise standard deviation $\sigma \leq c\eps/\sqrt[4]{d\log n}$,
where $c>0$ is a sufficiently small constant.  
Then \whp the nearest neighbor of $\tilde q$ (in $\tilde P$) is $\tilde p^*$.
\end{lemma}
\begin{proof}
Write $(X_1,\ldots,X_d)^\tran  = \tilde{q}-q \sim N_d(0,\sigma I_d)$, 
and similarly $(Y_1,\ldots,Y_d)^\tran = \tilde{p}^*-p \sim N_d(0,\sigma I_d)$.
Let $Z_i = X_i - Y_i$, and note that the $Z_i$'s are independent
Gaussians, each with mean $0$ and variance $2 \sigma^2$.
Then by direct computation
\begin{equation}
  \norm{\tilde{q} - \tilde{p}^*}^2 
  = 1  + \sum_{i=1}^d (X_i - Y_i)^2 + \sum_{i=1}^d (X_i- Y_i)(q_i - p^*_i) = 1+  \sum_{i=1}^d Z_i^2 + \sum_{i=1}^d  Z_i(q_i - p^*_i).
\end{equation}

For the term $\sum_{i=1}^d Z_i^2$,  Theorem
\ref{thm:XiSquare} gives us $\Pr \left[\left|\sum_{i=1}^d Z_i^2-2\sigma^2d\right| \geq 2
  \sigma^2 \left(x +d \cdot 2 \sqrt{\tfrac{x}{d}} \right)
  \right] \leq 2e^{-x}.$
Setting $x=4 \log n \leq O(d)$, observe that 
$ 2 \sigma^2(x+2\sqrt{xd}) 
  \leq O(\sigma^2\sqrt{xd}) 
  \leq O(c^2\eps^2)$,
and thus  we have
$\Pr\left[\left|\sum_{i=1}^d Z_i^2-2 d \sigma^2\right| \geq O(c^2\epsilon^2)\right] \leq \tfrac{2}{n^4}$.

Now the term $\sum_{i=1}^d Z_i(q_i - p^*_i)$ is a Gaussian with mean $0$ and variance
$\sum_{i=1}^d (q_i - p^*_i)^2 \var[Z_i] = 2\sigma^2 \norm{q-p^*}^2 = 2\sigma^2$,
and thus with high probability
$\abs{ \sum_{i=1}^d Z_i(q_i - p^*_i) } \leq O( \sigma \sqrt{\log n})$.
Substituting for $\sigma$ and recalling $d=\Omega(\log n)$, 
the righthand-side can be bounded by $O( c\epsilon )$. 
Altogether, with high probability
\begin{equation}\label{near}
  \norm{ \tilde{q} - \tilde{p}^* }^2 
  \le 1 + 2 d \sigma^2 \pm O( c^2\epsilon^2) \pm O(c\epsilon).
\end{equation}
Similarly, for every other point $p \neq p^*$, with high probability,
\begin{equation}\label{far}
  \norm{ \tilde{q} - \tilde{p} }^2 
  \geq \norm{q-p}^2 + 2 d \sigma^2 \pm O(c^2\epsilon^2) 
       \pm O(c\epsilon) \norm{q-p},
\end{equation}
and we can furthermore take a union bound over all such points $p\neq p^*$.
Now comparing Eqns.~\eqref{near} and \eqref{far}
when $\norm{q-p}^2\geq 1+\eps$ and $c>0$ is sufficiently small, 
gives us the desired conclusion.
\end{proof}

\begin{remark}
\label{rem:noisePerpendicular}
The problem remains essentially the same if we assume
the noise has no component in the space $U$.
Indeed, we can absorb the noise inside $U$ into the ``original'' points ($P$ and $q$). 
With high probability, this changes the distance from $q$ to every point in $P$ by at most $O(\sigma\sqrt{k\log n})\ll \eps$. 
Hence, in the rest of the article, we will assume the noise is
perpendicular to $U$.
\end{remark}


\section{Warmup: Iterative PCA under Small Adversarial Noise} 
\label{sec:bounded}

To illustrate the basic ideas in our ``iterative PCA'' approach, we
first study it in an alternative, simpler model that differs from Section
\ref{sec:model} in that the noise is \emph{adversarial} but of {\em small
magnitude}.  The complete ``iterative PCA'' algorithm for the model
from Section \ref{sec:model} will appear in Section
\ref{sec:largeGaussian}.

In the bounded noise model, for fixed $\eps\in(0,1)$,
we start with an $n$-point dataset $P$ and a point $q$,
both lying in a $k$-dimensional space $U \subset \R^d$,
such that 
\begin{equation} \label{eq:pstarAdv}
  \exists p^*\in P \text{ such that }
  \norm{q-p^*}\leq 1 \text { and }
  \forall p\in P\setminus\{p^*\},\ 
  \norm{q-p} \geq 1+\eps
\end{equation}

The set $\tilde P$ consists of points $\tilde p_i=p_i+t_i$ for all $p_i\in P$,
where the noise $t_i$ is arbitrary, 
but satisfies $\norm{t_i} \leq \eps/16$ for all $i$.
Similarly, $\tilde q=q+t_q$ with $\norm{t_q} \le \eps/16$.

\begin{theorem}\label{thm:boundedNoise}
Suppose there is a $(1+\eps/4)$-approximate NNS data structure for $n$
points in a $k$-dimensional Euclidean space with query time $\Fquery$,
space $\Fspace$, and preprocessing time $\Fprep$.  Then for the above
adversarial-noise model, there is a data structure that preprocesses
$\tilde P$, and on query $\tilde q$ returns $\tilde p^*$.  This data
structure has query time $O((dk+\Fquery)\log n)$, space $O(\Fspace)$,
and preprocessing time $O(n+d^3+\Fprep)$.
\end{theorem}

First we show that the nearest neighbor ``remains'' $p^*$ even after
the perturbations (similarly to Lemma \ref{lem:NNSstays}. Let
$\alpha=\eps/16$.

\begin{claim}\label{cl:boundednearest}
The nearest neighbor of $\tilde{q}$ in $\tilde P$ is $\tilde{p}^*$.
\end{claim}
\begin{proof}
For all $i$, we have $\|\tilde{p_i} - p_i\| \leq \norm{t_i} \leq \alpha$, 
hence by the triangle inequality, 
$\|\tilde{q}-\tilde{p}^*\| 
 \leq \|\tilde{q} - q\| + \norm{q - p^*} + \|p^* - \tilde{p}^*\| 
 \leq \norm{q - p^*} + 2\alpha$. 
For all $p \neq p^*$, a similar argument gives
$\|\tilde{q} - \tilde{p}\| 
 \geq \norm{q - p^*} + \eps - 2\alpha$.
\end{proof}

We now describe the algorithm used to prove Theorem
\ref{thm:boundedNoise}.  Our algorithm first finds a small collection
$\mathcal{U}$ of $k$-dimensional subspaces, such that every point of
$\tilde{P}$ is ``captured well'' by at least one subspace in
$\mathcal{U}$.  We find this collection $\mathcal{U}$ by iteratively
applying PCA, as follows (see Algorithm \ref{alg:iterPCA1}).  First
compute the top (principal) $k$-dimensional subspace of $\tilde P$.
It ``captures'' all points $\tilde p \in \tilde P$ within distance
$\sqrt{2}\alpha$ from the subspace. Then we repeat on the remaining
non-captured points, if any are left.
In what follows, let $p_{\tilde{U}}$ denote the projection of a point
$p$ onto $\tilde{U}$, and define the distance between a point $x$ and
a set (possibly a subspace) $S$ as $d(x,S) = \inf_{y\in S}
\norm{x-y}$.

\begin{algorithm} 
  \caption{Iteratively locate subspaces}
  \label{alg:iterPCA1}
  \begin{algorithmic}
  \STATE $j \gets 0$; $\tilde P_0 \gets \tilde{P}$
\WHILE{$\tilde P_j \neq \emptyset$}
  \STATE $\tilde{U}_j \gets$ the $k$-dimensional PCA subspace of $\tilde{P}_j$
  \STATE $M_j \gets \aset{ \tilde{p} \in \tilde{P}_j:\ d(\tilde{p}, \tilde{U}_j ) \leq \sqrt{2} \alpha }$
 \STATE $\tilde P_{j+1} \gets \tilde{P}_{j} \setminus M_j$
 \STATE $j \gets j+1$
  \ENDWHILE
  \RETURN $\mathcal{\tilde{U}} = \{\tilde{U}_0, \ldots,\tilde{U}_{j-1}  \}$ 
  and the associated point sets $\aset{M_0, M_1, \ldots, M_{j-1}}$.
  \end{algorithmic}
\label{alg:bounditer}
\end{algorithm}

The remainder of the preprocessing algorithm just constructs 
for each subspace $\tilde U\in {\mathcal U}$
a data structure for $k$-dimensional NNS, whose dataset is 
the points captured by $\tilde U$ projected onto this subspace $\tilde U$
(treating $\tilde U$ as a copy of $\R^k$).
Overall, the preprocessing phase comprises of $O(\log n)$ PCA computations 
and constructing $O(\log n)$ data structures for a $k$-dimensional NNS.

The query procedure works as follows.  Given a query point
$\tilde{q}$, project $\tilde{q}$ onto each $\tilde{U} \in \mathcal{U}$
to obtain $\tilde{q}_{\tilde{U}}$, and find in the data structure
corresponding to this $\tilde U$ a $(1+\eps/4)$-approximate nearest
neighbor point $\tilde{p}_{\tilde{U}}$ to $\tilde{q}_{\tilde{U}}$.
Then compute the distance between $\tilde{q}$ and each $\tilde{p}$
(original points corresponding to $\tilde p_{\tilde{U}}$), and report
the the closest one to $\tilde{q}$.

We now proceed to analyze the algorithm.

\begin{claim}
Algorithm~\ref{alg:bounditer} terminates within $O(\log n)$ iterations.
\end{claim}
\begin{proof}
Let $U$ be the PCA subspace of $P$ and let $\tilde{U}$ be the PCA
subspace of $\tilde{P}$. 
Since $\tilde U$ minimizes (among all $k$-dimensional subspaces)
the sum of squared distances from all $\tilde p\in \tilde{P}$ to $\tilde U$,
\[
  \sum_{\tilde p\in \tilde{P}} d(\tilde p,\tilde U)^2
  \leq \sum_{\tilde p\in \tilde{P}} d(\tilde p,U)^2
  \leq \sum_{\tilde p\in \tilde{P}} \norm{\tilde p-p}^2
  \leq \alpha^2 n.
\]
Hence, at most half of the points in $\tilde{P}$ may have distance to
$\tilde{U}$ which is greater than $\sqrt{2}\alpha$. The current set
$M$ will capture the other (at least a half fraction) points, and the
algorithm then proceeds on the remaining set.  Each
subsequent iteration thus decreases the number of points by a constant
factor.  After $O(\log n)$ iterations all points of $\tilde P$ must be
captured.
\end{proof}


\begin{claim} \label{cl:boundedReport}
The data structure for the subspace $\tilde U$ that captures $\tilde{p}^*$ always reports this point as the $(1+\eps/4)$-approximate nearest neighbor of $\tilde q$ (in $\tilde U$).
\end{claim}

We prove Claim \ref{cl:boundedReport} in Appendix \ref{app:bounded}.
The proof has to overcome the disparity between 
projection onto $U$, in which $p^*$ is the nearest neighbor, 
and onto the subspace $\tilde U$ used by the algorithm.
We achieve this by careful applications of the triangle inequality
and Pythagoras' Theorem, using the bounds on the noise-magnitude bound $\alpha<\eps/16$ and on the distance to the subspace $\sqrt{2}\alpha$.

We can now complete the proof of Theorem~\ref{thm:boundedNoise}.
By Claim~\ref{cl:boundedReport}, $\tilde{p}^*$ is always reported 
by the $k$-dimensional data structure it is assigned to.
But this is the closest point overall, by Claim~\ref{cl:boundednearest}, 
and thus our algorithm eventually reports this point $\tilde p^*$,
which proves the correctness part of Theorem \ref{thm:boundedNoise}. To
argue the time and space guarantees, we just note that computing one PCA
on $n$ points takes time $O(n+d^3)$, and there are in total $O(\log
n)$ PCAs to compute, and obviously also $k$-dimensional NNS data structures
to query against.

\section{Stability of a Top PCA Subspace}
\label{sec:sin-theta}

Before continuing to the full iterative-PCA algorithm, we need to address the challenge of 
controlling the stability of the
PCA subspace under random noise. In particular, we will need to show
that the PCA subspace $\tilde{U}$ computed from the noisy dataset
$\tilde{P}$ is ``close'' to the original subspace $U$.  We establish
this rigorously using the sine-theta machinery developed by Davis and
Kahan~\cite{DK70} and by Wedin~\cite{wedin}.

\paragraph{Notation}
Throughout, $s_j(M)$ denotes the $j$-th largest singular value
of a real matrix $M$, and $\norm{M}=s_1(M)$ denotes its spectral norm,
while $\norm{M}_F$ denotes the Frobenius norm of $M$.
All vector norms, i.e. $\norm{v}$ for $v\in\R^d$, refer to the $\ell_2$-norm.
We provide a more self-contained quick review of basic matrix analysis 
and spectral properties in Appendix \ref{apx:spectral}. 

\subsection{Wedin's $\sin\theta$ Theorem}

The \emph{$\sin\theta$ distance} between two subspaces $B$ and $A$ of
$\reals^d$ is defined as
$$
\sin \theta(B, A) = \max_{x \in B, \norm{x}=1}\, \min_{y \in A} \norm{x-y}. 
$$
Observe that the minimum here is just the distance to a subspace $\dist(x,A)$,
and it is attained by orthogonal projection.
Thus, for all $x'\in B$ (not necessarily of unit length)
$\dist(x',A) 
 = \norm{x'}\cdot \dist \left(\frac{x'}{\norm{x'}},A \right) 
 \leq \norm{x'}\cdot \sin\theta(B,A)$.

\begin{figure}[H]
  \begin{center}
    \includegraphics[scale = 0.4]{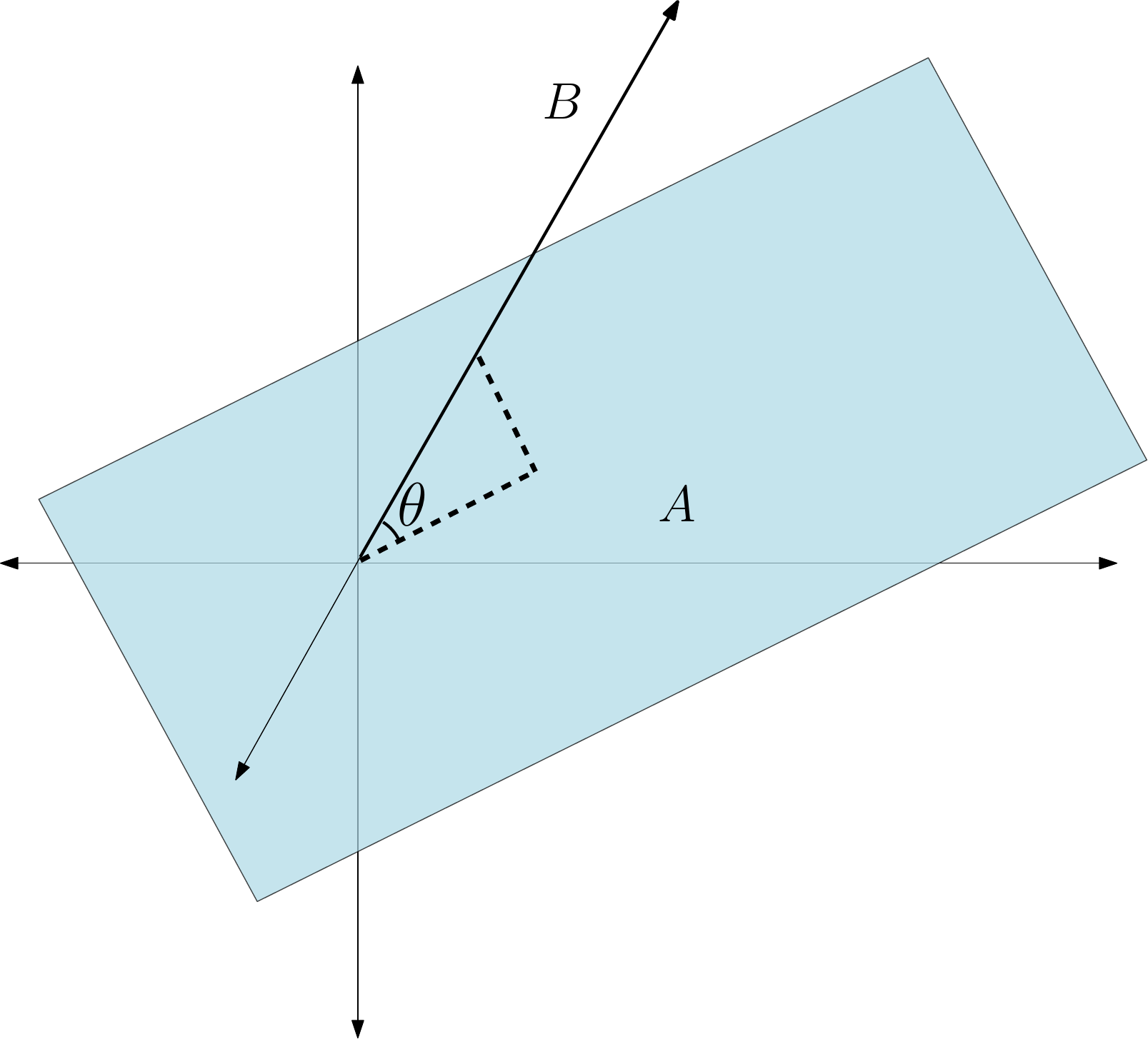}
  \end{center}
  \label{fig:sinthetaexample}
\end{figure}

For a matrix $X\in \R^{n \times d}$ and an integer $m\in\aset{1,\ldots,d}$, 
let $R_m(X)$ (resp. $L_m(X)$) denote the matrix formed by the
top $m$ right (resp. left) singular vectors of $X$ taken in column
(resp. row) order, and define $SR_m(X)$ (resp. $SL_m(X)$ ) as the
subspace whose basis is these right (resp. left) singular vectors. 

Now consider a matrix $X\in \R^{n \times d}$, and add to it a
``perturbation'' matrix $Y \in \R^{n\times d}$, writing $Z = X + Y$.
The theorem below bounds the $\sin\theta$ distance between the top
singular spaces before and after the perturbation, namely the
subspaces $SR_m(Z)$ and $SR_k(X)$ for some dimensions $m$ and $k$, 
in terms of two quantities:
\begin{enumerate} \compactify
\item
The projection of the perturbation $Y$ on $SR_m(Z)$ and on $SL_m(Z)$.
Let $Y_R = \norm{Y  R_m(Z)}$ and $Y_L = \norm{L_m(Z)  Y^\tran}$. 
\item
The gap between the top $m$ singular values of $Z$ and the bottom
$d-k$ singular values of $X$. Formally, define $ \gamma =s_m(Z) -
s_{k+1}(X)$.
\end{enumerate}

\begin{theorem}[Wedin's $\sin\theta$ Theorem \cite{wedin}]
\label{thm:sintheta}
In the above setting, if $m \leq k \leq d$ and $\gamma > 0$, then
\begin{equation*}
\sin \theta(SR_m(Z), SR_k(X)) \leq \frac{\max\aset{Y_R, Y_L}}{\gamma}.
\end{equation*}
\end{theorem}

\subsection{Instantiating the $\sin\theta$ Theorem}\label{perturb}

We now apply the $\sin\theta$-Theorem to our semi-random model
from Section~\ref{sec:model}. 
Let $X \in \reals^{n \times d}$ be the matrix corresponding to our original point set $P$ (of size $n \geq d$) lying in a subspace $U$ of dimension $k \leq d$. Let $T\in\reals^{n\times d}$ be a perturbation matrix (noise),
and then $\tX = X +T$ corresponds to our perturbed point set $\tilde P$. 
Our next theorem uses $\norm{T}$ directly
without assuming anything about its entries, 
although in our context where the entries of $T$ 
are drawn from independent Gaussians of magnitude $\sigma$,
Theorem~\ref{randomtheory} implies that \whp $\norm{T} \leq O(\sigma \sqrt{n+d})$.
In fact, if the matrix $T$ is random, 
then $m$ (and possibly also $\gamma$) should be interpreted as 
random variables that depend on $T$.

\begin{theorem}\label{evendist}
Let $\tX=X+T$ be defined as above, and fix a threshold $\gamma_1>0$.
If $m\leq k$ is such that at least $m$ 
singular values of $\tX$ are at least $\gamma_1$,
then 
\begin{equation*}
  \sin \theta(SR_m(\tX), SR_k(X)) \leq 
  \frac{\norm{T}}{\gamma_1},
\end{equation*}
where $SR_k(M)$ denotes, as before, 
the span of the top $k$ right-singular vectors of a matrix $M$. 
\end{theorem}

\begin{proof}
Towards applying Theorem~\ref{thm:sintheta},
define $T_R = \|T  R_m(\tX)\|$ and $T_L = \| L_m(\tX)  T^\tran \|$.
The columns of $R_m(\tX)$ being orthonormal implies $\norm{R_m(\tX)} \leq 1$,
and now by Fact \ref{tria},
$
  T_R = \norm{T R_m(\tX)} 
  \leq \norm{T}
$.
We can bound also $T_L$ similarly. 
Recalling Fact~\ref{fa:sumsingular}, 
$X$ has at most $k$ non-zero singular values
because the point set $P$ lies in a $k$-dimensional subspace,
hence the gap is $\gamma=s_m(\tilde X)-0 \geq \gamma_1$.
Plugging this into the $\sin\theta$ Theorem yields the bound
$\sin \theta(SR_m(\tX), SR_k(X)) \leq \norm{T}/\gamma \leq \norm{T}/\gamma_1$.
\end{proof}

\section{Iterative PCA Algorithm}\label{sec:largeGaussian}

We now present the iterative PCA algorithm, that solves the NNS
problem for the semi-random model from Section~\ref{sec:model}.  In
particular, the underlying pointset lives in a $k$-dimensional space,
but each point is also added a Gaussian noise $N_d(0,\sigma^2 I_d)$,
which has norm potentially much larger than the distance between a
query and it nearest neighbor. The algorithm reduces the setting to a
classical $k$-dimensional NNS problem.

\begin{theorem}\label{thm:mainfour}
Suppose there is a $(1+\eps/8)$-approximate NNS data structure for $n$
points in a $k$-dimensional space with query time $\Fquery$, space
$\Fspace$, and preprocessing time $\Fprep$.  Assume the Gaussian-noise
model \eqref{eq:tildeP}-\eqref{eq:pstar}, with 
$ \sigma (k^{1.5} \sqrt{\log n} + \sqrt[4]{k^3 d \log n}) < c\eps$
for sufficiently small constant $c>0$. 

Then there is a data structure that preprocesses $\tilde P$, and on
query $\tilde q$ returns $\tilde p^*$ with high probability.  This
data structure has query time $O((dk+\Fquery)\sqrt{d \log n} + d^{O(1)})$,
uses space $O(\Fspace \sqrt{d \log n}+d^{O(1)})$, and preprocessing time
$O((nd^2+d^3+\Fprep)\sqrt{d \log n})$.
\end{theorem}

\subsection{Algorithm Description}

The iterative-PCA algorithm computes a collection $\mU$ of
$O(\sqrt{d\log n})$ subspaces, such that every point in the perturbed
dataset $\tilde{P}$ is within squared distance $\Psi = d \sigma^2
+0.001\eps^2$ of some subspace in the collection $\mathcal{U}$.  For each
such subspace $\tU_j^s\in \mU$, we project onto $\tU_j^s$ the points
captured by this subspace $\tU_j^s$, and construct on the resulting
pointset a $k$-dimensional NNS data structure. We consider only
singular vectors corresponding to sufficiently large singular values,
which helps ensure robustness to noise. In particular, this threshold is 
$\delta(n) \triangleq c  \eps \sqrt{ \tfrac{n}{ k}}$ for small constant $c \leq 0.001$. Also, the PCA space is
computed on a sample of the current pointset only.

See Algorithm~\ref{iter} for a detailed description of computing
$\mU$.

\begin{algorithm}
  \caption{Iteratively locate subspaces. }
  \begin{algorithmic}
  \STATE Define $\Psi \triangleq d\sigma^2 + 0.001 \eps^2$,\,  
  $r\triangleq O(d^9k^3\tfrac{\log n}{\eps^2\sigma^2})$,\, 
  and $\delta(n) \triangleq c  \eps \sqrt{ \tfrac{n}{ k}}$ for a small constant $c \leq 0.001$.
  \STATE $j \gets 0$, $\tilde P_0 \gets \tilde{P}$
\WHILE{$|\tilde P_j| > r$}
 \STATE Sample $r$ points from $\tilde P_j$ (with repetition) to form the set/matrix $\tilde P_j^s$
 \STATE $m \gets$  number of singular values of $\tilde{P}_j^s$ that are at least $\delta(r)$

  \STATE $\tilde{U}_j^s \gets$ the subspace spanned by the $m$ top singular vectors of $\tilde P_j^s$
  \STATE $M_j \gets$ all $\tilde{p} \in \tilde{P}_j\setminus \tilde P_j^s$ at distance $\dist(\tilde{p}, \tilde{U}_j^s ) \leq \sqrt{\Psi}$
 \STATE $\tilde P_{j+1} \gets \tilde{P}_{j} \setminus (M_j\cup \tilde{P}_j^s) $
 \STATE $j \gets j+1$
\ENDWHILE
  \RETURN the subspaces $\mathcal{\tilde{U}} = \{\tilde{U}_0^s, \ldots,\tilde{U}_{j-1}^s  \}$, their pointsets 
$\{M_0, M_1, \ldots, M_{j-1} \}$, and the remaining set $R=\tP_{j}\bigcup \cup_{l=0}^{j-1} \tilde P_{l}^s$.
  \end{algorithmic}
\label{iter}
\end{algorithm}

We now present the overall NNS algorithm in detail.  The preprocessing
stage runs Algorithm~\ref{iter} on the pointset $\tP$, stores its
output, and constructs a $k$-dimensional NNS data structure for each
of the pointsets $M_0,\ldots,M_{j-1}$ (here $j$ refers to the final
value of this variable). Note that we also have a ``left-over'' set
$R=\tP_{j}\bigcup \cup_{l=0}^{j-1} \tilde P_{l}^s$, which includes the
points remaining at the end plus the sampled points used to
construct the subspaces $\mathcal U$.  

The query stage uses those $j$ data structures to compute a
$(1+\eps/8)$-approximates NNS of $q$ in each of $M_0,\ldots,M_{j-1}$,
and additionally finds the NNS of $q$ inside $R$ by exhaustive search.
It finally reports the closest point found.  In the rest of this
section, we will analyze this algorithm, thus proving Theorem
\ref{thm:mainfour}.

We make henceforth three assumptions that hold without loss of
generality.  First, we assume that $\norm{p_i} \geq 1$, which is
without loss of generality as we can always move the pointset away
from the origin. Overall, this ensures that $\|
P\|_F^2\ge |P|$.

Second, we assume that all points $\tilde P$
have norm at most $L\triangleq d^{3/2}$, 
which follows by applying a standard transformation of partitioning
the dataset by a randomly shifted grid with side-length $d$. This
transformation ensures that the query and the nearest neighbor, at
distance $O(\sigma\sqrt{d})$ are in the same grid cell with
probability at least $1-o(1)$ (see, e.g., \cite{highdreport}). 

Third, we assume that $\sigma\gg \eps/\sqrt{d}$, as otherwise we can
apply the algorithm from Section \ref{sec:bounded} directly. (The
algorithm in the current section works also for small $\sigma$,
but the exposition becomes simpler if we assume a larger $\sigma$.) In the context 
of our model of Section \ref{sec:model} and Lemma \ref{lem:NNSstays}, this is equivalent to asserting $d \gg \log n$.

Finally, we remark that the algorithm can be changed to not use
explicitly the value of $\sigma$, by taking only the closest $O
\left(\sqrt{\tfrac{\log n}{d}} \right)$ fraction of points to a given
space $\tilde U_j^s$. We omit the details.

\subsection{Analysis}

We now present a high-level overview of the proof. First, we
characterize the space $\tilde U_j^s$, and in particular show that it
is close to (a subspace of) the original space $U$, using the
sine-theta machinery and matrix concentration bounds. Second, we use
the closeness of $\tilde U_j^s$ to $U$ to argue that: (a) projection of
the noise onto $\tilde U_j^s$ is small; and (b) the projection of a
point $\tilde p$ is approximately $\|p\|$, {\em on average}. Third, we
use these bounds to show that the space $\tilde U_j^s$ captures a good
fraction of points to be put into $M_j$, thus allowing us to bound the
number of iterations. Fourth, we show that, for each point $\tilde
p=p+t$ that has been ``captured'' into $M_j$, its projection into
$\tilde U_j^s$ is a faithful representations of $p$, in the sense
that, for such a point, the distance to the projection of $\tilde q$
onto $\tilde U_j^s$ is close to the original distance (before
noise). This will suffice to conclude that the $k$-dimensional NNS for
that set $M_j$ shall return the right answer (should it happen to have
the nearest neighbor $p^*$).

Slightly abusing notation, let $P$ represent both the pointset and the
corresponding $n\times d$ matrix, and similarly for $\tilde{P}$ or a
subset thereof like $\tP_j$.  Let $T$ be the noise matrix, i.e., its
rows are the vectors $t_i$ and $\tilde{P} = P + T$.

Using bounds from random matrix theory (see Lemma~\ref{lem:subsetNorm}),
w.h.p. every restriction of $T$ to a subset of at least $d$ rows gives a submatrix $T'$ of spectral norm 
$\|T'\|\le \eta(|T'|)=O(\sigma\sqrt{|T'|\cdot \log n})$.  
In addition, by Corollary \ref{cor:projec} and the parameters of our model in Theorem \ref{thm:mainfour}, \whp
\begin{equation} \label{eq:noisebound}
  \forall p_i\in P, \qquad
  | \norm{t_i}^2 - \sigma^2 d | \leq 0.0001\eps^2.
\end{equation}
We assume in the rest of the proof that these events occur. 
Since both are high probability events, we may use a union bound 
and assume they occur over all iterations without any effect of conditioning on the points.

The thrust of our proof below is to analyze one iteration of
Algorithm~\ref{iter}.  We henceforth use $j$ to denote an arbitrary
iteration (not its final value), and let $n_j=\card{\tP_j} > r$ denote
the number of points at that iteration.

\subsubsection{Analysis: Characterization of the PCA space of the sampled set}
Define $\tilde U_j^s$ and $\tilde U_j$ to be the PCA space of $\tilde
P_j^s$ and $\tilde P_j$ respectively, i.e., full current set and
sampled set. Suppose the dimension of $\tilde U_j^s$ and $\tilde U_j$
is $m \leq k$ and $m \leq \ell \leq k$ respectively where $m$ is set
according to the thresholding step in Algorithm \ref{iter} and $\ell$
will be specified later. We show that the computed PCA space
$\tilde{U}_j^s$ is close to $U$ using the sine-theta machinery
established in Section~\ref{perturb}. We consider the point sets as
matrices, and concatenate the following two observations for deriving
our result:
\begin{itemize}
 \item{The PCA space of sampled noisy set (scaled) is close to that of the full noisy set.}
\item{The PCA space of the noisy set is close to that of the unperturbed set.}
\end{itemize}
We now analyze the effects of sampling. We sample $r$ points from the
current set $\tilde P_j$ of size $n_j$. We use the following standard
matrix concentration.

\begin{theorem}[Rudelson and Vershynin \cite{rudelson2007sampling}, Thm 3.1]
Suppose we sample (with replacement) $r$ row vectors from an $n$-size set $A\subset
\R^d$ (represented as a $n\times d$ matrix), and call them set
$Y=\{y_1,\ldots y_r\}$. Then, for any $t\in(0,1)$, and $L=\max_{a\in A} \|a\|$:
$$
\Pr\left[\left\|\frac{n}{r}{\sum_{y\in Y} y^\tran y -A^\tran A} \right \|>t \|A^\tran A \|\right]\le 2e^{-\Omega \left( \left(t^2 /L^2 \right) \cdot r/\log r \right)}.
$$
\end{theorem}

\begin{corollary}
\label{col:samplingConcetration}
For $t\in (0,1)$ if we sample $r = O(\log^2 n \cdot L^2/t^2)$ vectors from
$\tilde P_j$ to obtain $\tilde P_j^s$, we have that w.h.p.:
$$
\|\tfrac{n_j}{r}(\tilde P_j^s)^\tran \tilde P_j^s - \tilde P_j^\tran \tilde P_j\|\le L^2n_j\cdot t.
$$
\end{corollary}

\begin{proof}
Instantiate the theorem for $A = \tilde P_j$ to obtain
$\|\tfrac{n_j}{r}(\tilde P_j^s)^\tran \tilde P_j^s - \tilde
P_j^\tran \tilde P_j\|\le t \|\tilde P^\tran \tilde P \|$. We simplify
the right hand side of this equation. First, by Fact \ref{fa:co} we
have that $t \| \tilde P_j^\tran \tilde P_j \| = t \| \tilde P_j \|^2
$. Next by Fact \ref{fa:sumsingular}, $t \|\tilde P_j \|^2
= \sum_{\tilde p \in \tilde P_j} \|\tilde p \|^2 t \leq
n_j \max_{\tilde p \in \tilde P_j} \|\tilde p \|^2 t \leq L^2 n_j t$.
To prove this event succeeds with high probability over all points
$n$, we need only substitute the value of $r$ directly into the
theorem statement.
\end{proof}

\begin{corollary}
\label{col:preciseConcentration}
We set $t=O \left( \tfrac{\eps\sigma \sqrt{\log n} }{ L^2k^{1.5}} \right)$, for which 
we need to sample  $r=\Omega(L^6 k^3 \log n /\eps^2 \sigma^2)=\Omega(d^9 k^3 \log n / \eps^2\sigma^2)$. Then we obtain:
$$
\|\tfrac{n_j}{r}(\tilde P_j^s)^\tran \tilde P_j^s - \tilde P_j^\tran \tilde P_j\|\le n_j\cdot O \left(\eps\sigma\tfrac{\sqrt{\log n}}{k^{1.5}} \right) 
\ll \left( \frac{\delta(n_j)}{ k} \right)^2,$$
 for $\sigma$ in the range given by our model of Theorem \ref{thm:mainfour}.
\end{corollary}

We now aim to show that $\sin \theta(\tilde U_j^s, \tilde U_j)$ and
$\sin \theta(\tilde U_j, U_j)$ are small, and the triangle inequality
for $\sin \theta$ will then show that $\sin \theta(\tilde U_j^s, U_j)$
is also small.  Recall first that the sine-theta machinery of
Theorem \ref{perturb} requires a ``gap" between the bottom most
singular value considered in one subspace and the top most not
considered in the other. We observe the following lemma:

\begin{lemma}\label{lem:gapinU}
There exists $\ell$, where $m \leq \ell \leq k+1$, such that $s_{\ell}(\tilde P_j) - s_{\ell+1}(\tilde P_j) \geq \Omega \left( 
\frac{\delta(n_j)}{k} \right)$. Hence $s_{\ell}(\tilde P_j) - s_{k+1}(P_j) \geq \Omega \left( 
\frac{\delta(n_j)}{k} \right)$.
\end{lemma}
\begin{proof}
First recall that by the threshold step of our algorithm,  $s_m \left( \tfrac{n_j}{r}(\tilde P_j^s)^\tran \tilde P_j^s \right)
\geq \delta(n_j)^2$. Now by Fact \ref{fa:difference} and $r$ set as in
Corollary \ref{col:preciseConcentration}, we have that $s_m(\tilde
P_j^\tran \tilde P_j) \geq \delta(n_j)^2 - \delta(n_j)^2 / k^2
\ge \tfrac{3}{4} \delta(n_j)^2$. Hence we have $s_m(\tilde P_j) >
3 \delta(n_j) / 4 $. Also since $P_j$ is drawn from a $k$ dimensional
subspace, $s_{k+1} (P_j) = 0$ and therefore
$s_{k+1}(\tilde{P}_j) \leq \norm{P_j - \tilde P_j} =
O(\sigma \sqrt{n_j \log n}) < \delta(n_j) / 16$ by
Lemma \ref{lem:subsetNorm} and the parameters of our model.  Now since
$s_m(\tilde P_j) \geq 3 \delta(n_j) / 4$ and $s_{k+1}(\tilde
P_j) \leq \delta(n_j) / 16$, then there must exist $\ell$ with
$m \leq \ell \leq k+1$ that satisfies the claim.
 See Figure \ref{fig:staggeredgap} for illustration of the argument.
\end{proof}

Using now that $a^2-b^2 \geq (a-b)^2$ for $a \geq b \geq 0$ we obtain:
\begin{corollary}\label{col:gapinU}
There exists $\ell$,  $m \leq \ell \leq k+1$, such that $s_{\ell} \left(\tilde{P_j}^\tran \tilde P_j \right) - s_{\ell+1}
\left(\tilde{P_j}^\tran \tilde P_j \right) \geq \Omega \left(  \frac{\delta(n_j)^2}{k^2}  \right)$. Hence 
$ s_{m} \left(\tfrac{n_j}{r} \left(\tilde P_j^s \right)^\tran \tilde P_j^s  \right)  -   s_{\ell+1}(\tilde P_j^\tran \tilde P_j) \geq \Omega \left(  \frac{\delta(n_j)^2}{k^2}  \right)$.
\end{corollary}

\begin{figure}
  \begin{center}
    \includegraphics[scale = 0.6]{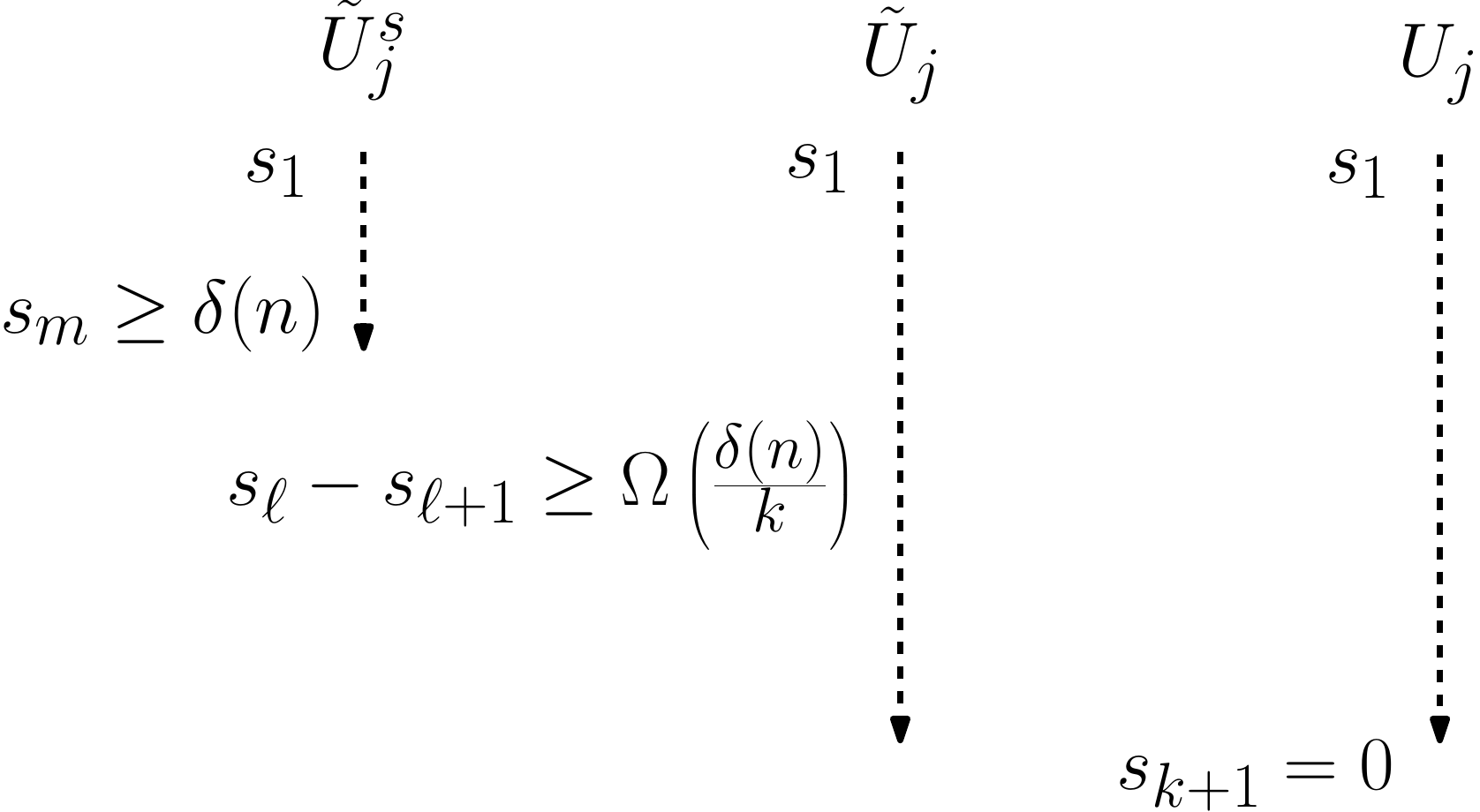}
  \end{center}
  \caption{Take $s_l$ as the last singular value included in $\tilde{U}_j$.}
  \label{fig:staggeredgap}
\end{figure}

Now crucially using this $\ell$ of Lemma \ref{lem:gapinU} to define our cut off point for the singular vectors included in $\tilde U_j$, and instantiating the sine-theta theorem, we obtain the following:

\begin{lemma}\label{lem:sinthetaTwice}
$\sin \theta(\tilde U_j^s , \tilde U_j) \leq O \left( \frac{\sigma k^{1.5} \sqrt{ \log n}}{\eps} \right)$ and $\sin \theta(\tilde U_j , U_j) \leq O \left( \frac{\sigma k^{1.5} \sqrt{ \log n}}{\eps} \right)$.
\end{lemma}

\begin{proof}
First by Lemma \ref{lem:gapinU}, we have  $s_l(\tilde P_j) - s_{k+1}(P_j) \geq O \left( \frac{\delta(n)}{k} \right)$ and recall by Lemma \ref{lem:subsetNorm}, $\norm{\tilde P_j - P_j} = O(\sigma \sqrt{n_j \log n})$. So instantiating the sine-theta theorem, 
Theorem \ref{evendist}:

\begin{align*}
\sin \theta(\tilde U_j, U_j) &=  \frac{\norm{\tilde P_j - P_j}  }{s_l(\tilde P_j) - s_{k+1} (P_j) } \\
&= O  \left(\frac{\sigma \sqrt{n_j  \log n} }{ \delta(n_j) / k} \right) \\
&= O \left(\frac{\sigma \sqrt{n_j \log n}}{ \eps \sqrt{n_j} / k^{1.5} }  \right)  = O \left( \frac{\sigma k^{1.5} \sqrt{ \log n}}{ \eps }  \right).
\end{align*}

Similarly for $\sin \theta(\tilde U_j^s, \tilde U_j)$, we upper bound  $\|\tfrac{n_j}{r}(\tilde P_j^s)^\tran \tilde P_j^s - \tilde P_j^\tran \tilde P_j\|$ by Corollary 
\ref{col:preciseConcentration}. We lower bound the ``gap" in singular values using Corollary 
\ref{col:gapinU}, and hence by instantiating the sine-theta theorem:

\begin{align*}
\sin \theta(\tilde U_j^s, \tilde U_j) &= O \left(\frac{ \eps n_j \sigma \sqrt{\log n} / k^{1.5}}{ \delta(n_j)^2/k^2} \right) \\
&= O \left( \frac{\eps n_j \sigma \sqrt{\log n} / k^{1.5}}{ \eps^2 n_j / k^3} \right) 
= O \left( \frac{ \sigma k^{1.5} \sqrt{\log n}}{\eps} \right). 
\end{align*}
\end{proof}

Since $\sin \theta$ is concave in the right regime, Lemma \ref{lem:sinthetaTwice} now gives us as a simple corollary the main claim of this subsection:

\begin{equation}\label{eq:overallSinTheta}
\sin \theta(\tilde U_j^s, U_j) \leq O \left(\frac{\sigma k^{1.5} \sqrt{\log n}}{\eps} \right).
\end{equation}

\subsubsection{Analysis: Noise inside the PCA space}

We now show that the noise vector $t_i$ of each point $\tilde
p_i=p_i+t_i$ has a small component inside $\tilde{U}_j^s$.  We use the
$\sin\theta$ bound in Eqn. \eqref{eq:overallSinTheta} for this. (Tighter analysis is possible directly via the randomness of the
vector $t_i$ on the first iteration, but conditioning of points selected together at each iteration of our algorithm makes this complicated at latter stages.)

Following Remark~\ref{rem:noisePerpendicular}, we shall assume that
noise $t_i$ is perpendicular to $U$.  Define $V\in \R^{d \times k}$ as
the projection matrix onto the space $U$, so e.g.\ $t_iV$ is the
zero vector, and define analogously $\tilde{V}_j^s \in \R^{d \times
  m}$ to be the projection matrix onto the $m$-dimensional space
$\tilde{U}_j^s$. 
\begin{lemma} \label{lem:noisebound}
 $\norm{t_i \tilde{V}_j^s} \leq \norm{t_i} \cdot \sin\theta(\tilde{U}_j^s,U)$.
\end{lemma}

\begin{proof}
Let $x$ be a unit vector in the direction of $t_i \tilde{V}_j^s$, 
namely, $x = \frac{t_i \tilde{V}_j^s}{\|t_i \tilde{V}_j^s \|}$.  
This implies $\|t_i \tilde{V}_j^s \|=  x^{\tran} t_i$. 
Now decompose $x \in \tilde{U}_j^s$ as $x= \sqrt{1 - \beta^2} u + \beta  v$ 
for unit vectors $u \in U$ and $v \perp U$ and some $\beta\geq0$.
Then $\beta=d(x,U) \leq \sin\theta(\tilde{U}_j^s, U)$,
and we conclude
\begin{equation*} 
  \|t_i \tilde{V}_j^s \| 
  = x^{\tran} t_i
  = \sqrt{1 - \beta^2}  u^{\tran} t_i   + \beta  v^{\tran} t_i
  =0 + \beta  v^{\tran} t_i
  \leq \norm{t_i}\cdot \sin\theta(\tilde{U}_j^s, U). 
\qedhere
\end{equation*}
\end{proof}

\begin{corollary}\label{cor:noisebound}
 
For every point $p_i$ and iteration $j$,
$ \norm{t_i \tV_j^s} 
  \leq O \left(\frac{1}{\eps} \sigma^2 k^{1.5} \sqrt{d \log n } \right)$.
\end{corollary}
\begin{proof}
Substitute into Lemma \ref{lem:noisebound} the bounds 
$\norm{t_i} \leq O (\sigma \sqrt{d } )$ from Eqn. \eqref{eq:noisebound} and 
from Eqn. \eqref{eq:overallSinTheta}, \\ $\sin\theta(\tilde{U}_j^s, U)   \leq O \left(\frac{1}{\eps} \sigma k^{1.5} \sqrt{ \log n} \right)$.
\end{proof}

We note that according to our model parameters in Theorem \ref{thm:mainfour}, this implies that $ \norm{t_i \tV_j^s}  \leq c \eps$ for a small constant $c$ that depends only on the choice of constant in our model, and we shall use this assumption henceforth.

\subsubsection{Analysis: Projection of the data into the PCA space}
\label{sec:compare}

We now show that the component of a data point $\tilde p_i$ inside the
PCA space $\tilde U_j^s$ of some iteration $j$, typically recovers
most of the ``signal'', i.e., the unperturbed version $p_i$.  More
precisely, we compare the length seen inside the PCA space
$\norm{\tilde{p}_i \tV_j^s}$ with the original length
$\norm{p_i}$. While the upper bound is immediate, the lower bound
holds only {\em on average}.

\begin{lemma}
\label{lem:rightside}
W.h.p., for all $\tilde p_i \in \tP_j$,\, 
$ \norm{\tilde{p}_i \tilde{V}_j^s }^2 - \norm{p_i}^2 
  \leq  d \sigma^2 + 0.0001 \eps^2$. 
\end{lemma}

\begin{proof}
Using Pythagoras' Theorem,
$ \norm{\tilde{p_i} \tilde{V} _j^s}^2 
  \leq \norm{\tilde{p_i} }^2 
  = \norm{p_i}^2 + \norm{ t_i }^2$,
and the lemma follows by the noise bound \eqref{eq:noisebound}.
\end{proof}

\begin{lemma}
\label{lem:leftside}
$ \sum_{\tilde p_i\in \tP_j} ( \norm{\tilde{p}_i \tV_j^s }^2 -  \norm{p_i}^2 ) 
  \geq  - k \delta(n_j)^2$. 
\end{lemma}
\begin{proof}
Let $V$ be the projection matrix into $U$ and $P_j$ the non-noised
version of $\tilde P_j$. Observe that $\tilde P_j V = P_j$, since the
noise is orthogonal to $U$.  Hence, by definition of PCA space
(Theorem~\ref{opt}):

\[
  \sum_{p_i\in P_j} \norm{p_i }^2 
  = \norm{ P_j }_F^2 
  = \norm{ \tilde{P}_j V }_F^2 
  \leq \sum_{l=1}^k s_l^2(\tilde{P_j}). 
\]
By Corollary \ref{col:preciseConcentration}, we further have that
$$
\tfrac{n_j}{r}\sum_{l=1}^k s_l^2(\tilde{P}_j^s) \ge \sum_{l=1}^k s_l^2(\tilde{P}_j) -  k\cdot \left( \tfrac{\delta(n_j)}{k} \right)^2 \geq \sum_i \norm{p_i }^2  -
k\cdot \left( \tfrac{\delta(n_j)}{k} \right)^2.$$
Or simply,
\begin{equation}\label{eq:firstside}
\tfrac{n_j}{r}\sum_{l=1}^k s_l^2(\tilde{P}_j^s) \ge \sum_i \norm{p_i }^2  -
k\cdot \left( \tfrac{\delta(n_j)}{k} \right)^2.
\end{equation}

 We also have:
\begin{align*}
\tfrac{n}{r}\sum_{l=1}^m s_l^2(\tilde{P}_j^s)
&=
\tfrac{n}{r}\|(\tilde{V}_j^s)^T (\tilde{P}_j^s)^\tran \tilde{P}_j^s \tilde{V}_j^s \|_F
=
\|(\tilde{V}_j^s)^T  \tfrac{n}{r} (\tilde{P}_j^s)^\tran \tilde{P}_j^s \tilde{V}_j^s \|_F \\
&=
\|(\tilde{V}_j^s)^T (\tilde{P}_j^\tran \tilde{P}_j + Z ) \tilde{V}_j^s \|_F
\le
\sum_{\tilde p_i\in \tilde P_j}\norm{\tilde{p}_i \tV_j^s }^2
+\|(\tilde{V}_j^s)^T Z \tilde{V}_j^s\|_F,
\end{align*}
where $Z$ has spectral norm at most $\delta(n_j)^2/k^2$, and hence
$\|(\tilde{V}_j^s)^T Z \tilde{V}_j^s\|_F\le \frac{k\delta(n_j)^2}{k^2} \le \frac{\delta(n_j)^2}{k}$.

Rearranging yields us:
$$
\tfrac{n}{r}\sum_{l=1}^m s_l^2(\tilde{P}_j^s) \le \sum_{\tilde p_i\in \tilde P_j}\norm{\tilde{p}_i \tV_j^s }^2 + \tfrac{\delta(n_j)^2}{k}.
$$

Finally, we have:
\begin{align*}
  \tfrac{n_j}{r}\sum_{l=1}^k s_l^2(\tilde{P}_j^s)  
  &= \tfrac{n_j}{r}\sum_{j=1}^m s_j^2(\tilde{P}_j^s)  + \tfrac{n_j}{r}\sum_{j=m+1}^k s_j^2(\tilde{P}_j^s) \\  
  &\leq \sum_{p_i\in \tilde P_j}\norm{\tilde{p}_i \tV_j^s }^2 + \tfrac{\delta(n_j)^2}{k} + \tfrac{n_j}{r}\sum_{j=m+1}^k s_j^2(\tilde{P}_j^s)   \\
 &\leq \sum_{p_i\in \tilde P_j}\norm{\tilde{p}_i \tV_j^s }^2 + \tfrac{\delta(n_j)^2}{k} + \tfrac{n_j}{r} (k-1) \delta(r)^2 \\
&\leq  \sum_{p_i\in \tilde P_j}\norm{\tilde{p}_i \tV_j^s }^2 + \tfrac{\delta(n_j)^2}{k} +  (k-1) \delta(n_j)^2 
\end{align*}
where we employed the threshold $s_j \leq \delta(r)$ for singular values taken by Algorithm~\ref{iter} with $j > m$, and that $\tfrac{n_j}{r}  \delta(r)^2 = \delta(n_j)^2$ by straightforward 
substitution of the formula $\delta(x) = \frac{c\eps \sqrt{x}}{ \sqrt{k}}$ .
The lemma follows by combining the above with Equation \ref{eq:firstside}.
\end{proof}

\subsubsection{Analysis: Number of iterations}

We now show that each iteration captures in $M_j$ a good fraction of
the remaining points, thereby bounding the number of iterations
overall. In particular, we give a lower bound on the number of indexes
$i$ such that $\tilde{p}_i$ is close to the $m$ dimensional PCA
subspace $\tilde{U}_j^s$, using results from Section
\ref{sec:compare}. Note that the square of this distance for a point
$\tilde{p}_i$ is precisely $\|\tilde{p}_i\|^2 - \|\tilde{p}_i
\tilde{V}\|^2$.  Let $X$ and $Y$ be quantities according to Lemmas
\ref{lem:rightside} and \ref{lem:leftside}, such that
\begin{align}
  & \norm{ \tilde{p}_i \tilde{V}_j^s }^2  - \|p_i \|^2  \leq  Y.   
\\
  Xn_j \leq & \sum_i (\norm{\tilde{p}_i \tilde{V}_j^s}^2  - \| p_i \|^2);
\end{align}

Now let $f$ be the fraction of $i$'s such that 
$\|\tilde{p}_i \tilde{V}_j^s\|^2 -\|p_i\|^2 \leq -0.0002 \eps^2$.
Then
\begin{equation}
Xn_j \leq \sum_i \|\tilde{p}_i \tilde{V}_j^s\|^2  - \sum_i \|p_i\|^2
\leq (1-f) n_j Y - 0.0002 f n_j  \eps^2.
\end{equation}
Rearrangement of terms gives us that $f \leq \frac{Y-X}{Y+ 0.0002 \eps^2}$.
By Lemma \ref{lem:leftside} , we can set $Xn_j = -k \delta^2 = - c^2 \eps^2
n_j = -0.00001 \eps^2 n_j$ and so $X \leq -0.00001 \eps^2$. And by Lemma \ref{lem:rightside}, we
have $Y \leq d \sigma^2 + 0.0001 \eps^2$. Elementary calculations now yield 
$$f \leq 1 - \Omega \left(\frac{\eps^2}{d
  \sigma^2} \right) \leq 1 - \Omega \left( \sqrt{ \frac{\log n}{d}}
\right).$$ 
(This last upper bound on $f$ can be made tighter as dependence on $\sigma$, but we opt for the looser bound which holds in our range of parameters for simplicity.) 
Now for the rest of the $(1-f)\ge \Omega\left(\sqrt{ \frac{\log n}{d}}\right)$
fraction of the points, the distance to the PCA subspace $\tilde{U}$
is, by Pythagoras' Theorem, 
\begin{align*}
  \|\tilde{p}\|^2 - \|\tilde{p} \tilde{V}_j^s\|^2  
  = \|p \|^2 + \|t \|^2 - \|\tilde{p} \tilde{V}_j^s\|^2 
  \leq \|t \|^2 + 0.0002 \eps^2.
\end{align*}
Since $\|t \|^2  \leq d \sigma^2 + 0.0001 \eps^2$, we get the required inequality that a large fraction of the points is within squared distance $d \sigma^2 + 0.001 \eps^2  = \Psi$.
It follows that the fraction of points captured by $\tilde{U}_j^s$,
i.e., in a single iteration,  
is at least $\Omega\left( \sqrt{\tfrac{\log n}{d}}\right)$,
which immediately implies a bound on the number of iterations, as follows.

\begin{lemma}
Algorithm \ref{iter} terminates in at most $ O \left(\sqrt{\frac{d}{\log n}}
\log n \right) = O(\sqrt{d \log n})$ iterations.
\end{lemma}

\subsubsection{Analysis: Correctness}
It now remains to show that the data structure that captures the
actual nearest neighbor $\tilde{p}^*$ will still report $\tilde{p}^*$
as the nearest neighbor to $\tilde q$ in the $k$-dimensional data
structure. Suppose $\tilde p^*$ has been captured in $j^{th}$
iteration, i.e., $\tilde p^*\in M_j$. For simplicity of exposition,
let $\tilde U=\tilde U_j^s$ and $\tilde V=\tilde V_j^s$.

Note that all distance computations for a query $\tilde{q}$ are of the
form $\|\tilde{q} \tilde{V} - \tilde{p} \tilde{V}\| =
\|(\tilde{q}-\tilde{p}){\tilde{V}}\|$, where $\tilde{p}$ is a point
that is close to $\tilde{U}$. Let $\tilde{q} = q + t_q$ and $\tilde{p}
= p + t$. Then we have for a point $\tilde{p}$:

\begin{align*}
\norm{(\tilde{q} - \tilde{p}) \tilde{V} } =  \norm{(q - p) \tilde{V} }  \pm  O \left( \norm{ t_q  \tilde{V} } +   \norm { t \tilde{V} }   \right)  
\end{align*}

Considering the noise of the query point $q$, by Corollary
\ref{cor:projec}, we have $\norm{ t_q \tilde{V} } \leq O(\sigma
\sqrt{k} + \sigma \sqrt{\log n})$ w.h.p. Similarly, considering the
noise $t$ of a point, by Corollary
\ref{cor:noisebound}, we have $\norm {t \tilde{V} } \leq O(\tfrac{\sigma^2}{\eps}
k^{1.5} \sqrt{d\log n})$. By the model specified in Theorem \ref{thm:mainfour}, we can set both these terms to be smaller than $0.01 \eps$. 
Hence we have:
\begin{equation}\label{eq:start}
\norm{ (\tilde{q} - \tilde{p}) \tilde{V} }=  \norm{ (q - p) \tilde{V} }  \pm  0.02 \eps.  
\end{equation}
Furthermore, we have $\sin \theta(\tilde U, U) \leq 0.01 \eps$.
We now decompose $U$ into two components. Let $U_{in}$ be the
projection of $\tilde{U}$ onto $U$, and $U_{out}$ be the space
orthogonal to $U_{in}$ but lying in $U$. See Figure \ref{fig:inandout}. 
We note that $U_{out}$ is
also orthogonal to $\tilde{U}$: otherwise some component of $U_{out}$
would lie in the projection of $\tilde{U}$ onto $U$, which is a
contradiction. Let $V_{in}$ and $V_{out}$ be the corresponding
projection matrices.  We likewise decompose each point $p$ as $p_{in}\in U_{in}$
and $p_{out}\in U_{out}$.

\begin{figure}[H]
  \begin{center}
    \includegraphics[scale = 0.4]{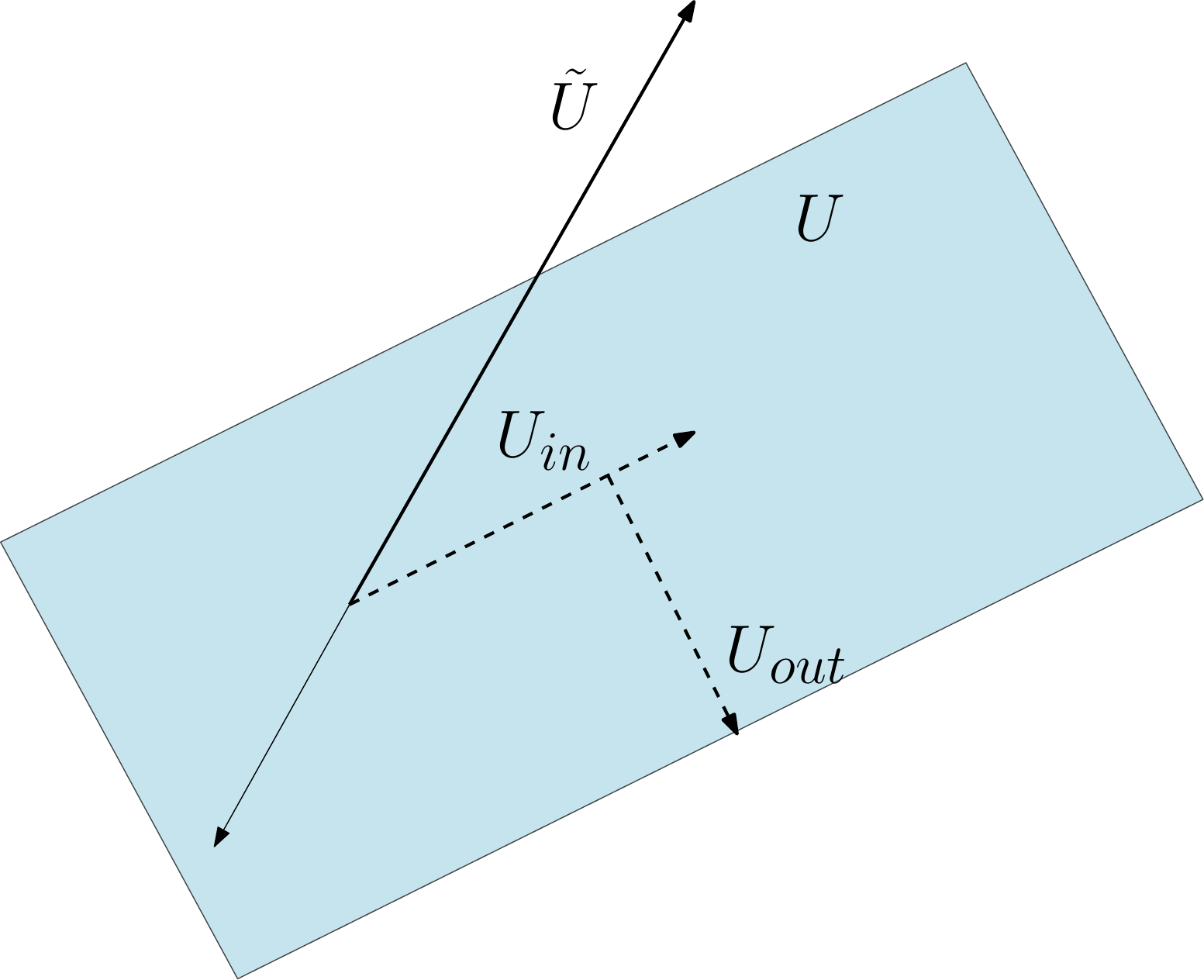}
  \end{center}
  \label{fig:inandout}
\caption{$U_{in}$ is projection of $\tilde U$ into $U$, and $U_{out}$ is the orthogonal complement of $U_{in}$ in $U$.}
\end{figure}

We make the following claim, which shows that bounding $\|p_{out}\|$
for all points in $M_j$ suffices for correctness of our algorithm.

\begin{lemma}\label{lem:smallout}
If $\norm{p_{out}} \leq 0.1\eps$ for all $p$ captured by subspace
$\tilde{U}$ (also capturing $p^*$), 
then $\tilde p^*$ remains a nearest neighbor to $\tilde q$
after projection to $\tilde{U}$. Also for any $p' \neq p^*$, we have that $\|(\tilde q - \tilde p') \tilde V \| \geq \left(1 + \tfrac{\eps}{8} \right) \|(\tilde q - \tilde p^*) \tilde V \|$.

\end{lemma}
\begin{proof}
Consider Equation \ref{eq:start}. 
\begin{align*}
\norm{ (\tilde{q} - \tilde{p}) \tilde{V} } &=  \norm{ (q - p) \tilde{V} }  \pm  0.02\eps
\\
&= \norm{(q_{in} - p_{in} + q_{out} - p_{out}) \tilde{V}} \pm 0.02\eps  \quad \text{(Decomposing $q$, $p$ into components in $U_{in}$ and $U_{out}$)}\\ 
&=  \norm{(q_{in} - p_{in}) \tilde{V} + (q_{out} - p_{out}) \tilde{V}} \pm 0.02\eps \\
&= \norm{(q_{in} - p_{in}) \tilde{V}}  \pm 0.02\eps \quad \text{(Since $U_{out} \perp \tilde{U}$)}
\\
&= (1  \pm 0.01\eps ) \norm{q_{in} - p_{in}} \pm 0.02\eps \quad \text{(Since $\sin \theta(U_{in}, \tilde{U}) \leq 0.01 \eps$)}
\end{align*}
To summarize these last calculations, we have:
\begin{equation}\label{eq:projbound}
\norm{ (\tilde{q} - \tilde{p}) \tilde{V} } = (1  \pm 0.01\eps ) \norm{q_{in} - p_{in}} \pm 0.02 \eps.
\end{equation}
Next note by Pythagoras:
\begin{equation}\label{eq:rearr}
\norm{q_{in} - p_{in}}^2 = \norm{q - p_{in}}^2 - \norm{q_{out}}^2.
\end{equation}
Also observe from the triangle inequality and the assumption of our lemma that $\norm{p_{out}} \leq 0.1\eps$, $\forall p \in P$ captured by the subspace, $\norm{q -
  p_{in}}=\norm{q - p} \pm \norm{p_{ out}} = \norm{q - p} \pm 0.1\eps$.

Hence, if $p^*$ is captured by the data structure, $\norm{q
  - p^*_{in}}^2 = (\norm{q-p^*} \pm 0.1\eps)^2 \leq
1 + 0.25 \eps$.  Similarly for $p' \neq
p^*$, we have $\norm{q - p'_{in}}^2 \geq (1+ \eps -0.1\eps)^2 \geq 1+1.8 \eps$. This gives:
\begin{align*}
 \frac{1 + 1.8 \eps }{ 1 + 0.25 \eps} \leq  \frac{\norm{q- p'_{in}}^2}{\norm{q - p^*_{in}}^2}  \leq \frac{\norm{q- p'_{in}}^2 - \norm{q_{out}}^2 }{\norm{q - p^*_{in}}^2   - \norm{q_{out}}^2} = \frac{\norm{q_{in}- p'_{in}}^2}{\norm{q_{in} - p^*_{in}}^2},
\end{align*}
where we crucially used in the second step the fact that
subtracting the same quantity from both numerator and denominator of a
fraction can only increase it. Some elementary algebraic manipulation shows then:

\begin{equation}\label{eq:inside}
\frac{\norm{q_{in}- p'_{in}}}{\norm{q_{in} - p^*_{in}}} \geq 1 + \frac{\eps}{4}.
\end{equation}

Now we lower bound $\norm{q_{in} -p'_{in}}$. We do so as follows:
\begin{align*}
  \norm{q_{in} -  p'_{in}}^2 - \norm{q_{in} - p^*_{in}}^2 = &\norm{q - p'_{in}}^2 -\norm{q - p^*_{in}}^2    &\text{   (By Equation \ref{eq:rearr}}) \\
 \geq &(1 + 0.9\eps)^2 - (1+0.1\eps)^2 \ge 1.6\eps.
\end{align*}
This implies:
\begin{equation}\label{eq:lblength}
 \norm{q_{in} -  p'_{in}} \geq 1.2 \eps
\end{equation}

 Finally, we come to our main claim of the ratio of  $ \norm{ (\tilde{q} -\tilde{p}') \tilde{V} }$ to $ \norm{ (\tilde{q} - \tilde{p}^*) \tilde{V} }$. By Equation \ref{eq:projbound}, this is at least:
$$
\frac{ \norm{ (\tilde{q} -
    \tilde{p}') \tilde{V} } }{ \norm{ (\tilde{q} - \tilde{p}^*)
    \tilde{V} }}
\ge \frac {(1 - 0.01\eps ) \norm{q_{in} -
    p'_{in}}- 0.02\eps}{(1 + 0.01\eps ) \norm{q_{in} -
    p^*_{in}} +0.02\eps}.$$
 Substituting the lower bound on $\frac{\norm{q_{in}- p'_{in}}}{\norm{q_{in} - p^*_{in}}}$ from Equation \ref{eq:inside}, and the lower bound on $\norm{q_{in}-p'_{in}}$ from Equation \ref{eq:lblength} completes our claim.
\end{proof}

Next we derive a sufficient condition for $\norm{ p_{out}}$ to be
bounded as desired. For a vector $a$, we define $a \tilde{V}$ and
$a^{\perp{\tilde{U}}}$ to be the components lying in and orthogonal to
subspace $\tilde{U}$ respectively.

The quantity of interest is $c_{p}$, defined as the cosine of
the angle between $(p_{in})^{\perp \tilde{U}}$ and $t^{\perp
  \tilde{U}}$, for a point $\tilde p=p_{in}+p_{out}+t$.

\begin{lemma}\label{lem:suff}
Decompose $p$ as $p = p_{in} + p_{out} + t$, where $t \perp U$. 
Suppose $c_p \leq C \frac{\eps}{\sigma \sqrt{d}} $ for
suitable choice of constant $C$.
Then $\norm{p_{out}} \leq 0.1\eps$ w.h.p.
\end{lemma}

\begin{proof}
 We show first that to upper-bound $p_{out}$, it suffices to lower bound $\norm{
  \left(p_{in} + t \right)^{\perp \tilde{U}} }^2$ by $d \sigma^2 -
0.009 \eps^2$. Indeed, for captured points, we have by construction:
\begin{equation}\label{eq:condA}
\norm{\tilde{p}}^2 - \norm{\tilde{p} \tilde{V}}^2 \leq \Psi
\end{equation}
where $\Psi=d \sigma^2 + 0.001 \eps^2$.

We convert the inequality to our desired condition as follows: 
\begin{align*}
& \norm{\tilde{p}}^2 - \norm{\tilde{p} \tilde{V}}^2 \leq 
 d \sigma^2 + 0.001 \eps^2   \\
 & \norm{p_{out}}^2 + \norm{p_{in} + t}^2 - \norm{\tilde{p} \tilde{V}}^2  \leq d \sigma^2 + 0.001 \eps^2  \\ 
& \text{(Since by Pythagoras, $\norm{\tilde{p}}^2 = \norm{p_{out}}^2 + \norm{p_{in} + t}^2$) }  \\
 & \norm{p_{out}}^2 + \norm{(p_{in} + t) \tilde{V}}^2  + \norm{(p_{in} + t)^{\perp \tilde{U}}}^2  - \norm{\tilde{p} \tilde{V}}^2  \leq d \sigma^2 + 0.001 \eps^2   \\
&\text{(Decomposing $p_{in} + t$ orthogonal to and lying in subspace $\tilde{U}$) } \\
  & \norm{p_{out}}^2 + \norm{\tilde{p} \tilde{V}}^2  + \norm{(p_{in} + t)^{\perp \tilde{U}}}^2  - \norm{\tilde{p} \tilde{V}}^2  \leq d \sigma^2 + 0.001 \eps^2   \\ 
&\text{(Since $p_{out} \perp \tilde{V}$ by construction, hence $\tilde{p} \tilde V = (p_{in}+ t )\tilde{V}$) } \\
 & \norm{p_{out}}^2  + \norm{(p_{in} + t)^{\perp \tilde{U}}}^2   \leq d \sigma^2 + 0.001 \eps^2  \\
 & \norm{p_{out}}^2  \leq   0.001 \eps^2 +  d \sigma^2  - \norm{(p_{in} + t)^{\perp \tilde{U}}}^2  
\end{align*}
Clearly now if $\norm{
  \left(p_{in} + t \right)^{\perp \tilde{U}} }^2 \geq d \sigma^2 -
0.009 \eps^2$, then $\norm{p_{out}}^2 \leq 0.01 \eps^2$ and would complete our proof.

We now show how the bound on $c_p$ implies the required lower bound on $\norm{
  \left(p_{in} + t \right)^{\perp \tilde{U}} }^2$.
First note by the law of cosines, that 
\begin{equation}\label{eq:sum}
\norm{ (p_{in} + t)^{\perp \tilde U} }^2 = \norm{p_{in}^{\perp \tilde U}}^2 + \norm{t^{\perp \tilde{U}}}^2 - 2 \norm{p_{in}^{\perp \tilde U}} 
\norm{t^{\perp \tilde U}} c_p.
\end{equation}

Next note that $\norm{t^{\perp \tilde{U}}}^2 = d \sigma^2 \pm 0.001 \eps^2$ w.h.p. and a suitably small constant $c$ in our bound on $\sigma$ in the model parameters.  
This follows from decomposing $\norm{t}^2 = \norm{t^{\perp \tilde{U}}}^2 + \norm{t \tilde{V}}^2$ 
by Pythagoras, the concentration on $\norm{t}^2$ by Equation \ref{eq:noisebound} and the upper bound on $\norm{t \tilde{V}}^2$ by Corollary \ref{cor:noisebound}.   We now solve to find the desired condition on $c_p$.

\begin{align*}
&\norm{p_{in}^{\perp \tilde U }}^2 + \norm{t^{\perp \tilde U}}^2 - 2 \norm{p_{in}^{\perp \tilde U}}
\norm{t^{\perp \tilde U}} c_p   \geq d \sigma^2 - 0.009 \eps^2  \\
&\norm{p_{in}^{\perp \tilde U }}^2 + (d \sigma^2  \pm 0.001 \eps^2) - 2 \norm{p_{in}^{\perp \tilde U}} 
\norm{t^{\perp \tilde U}} c_p   \geq d \sigma^2 - 0.009 \eps^2   \\
&\norm{p_{in}^{\perp \tilde U }}^2 - 2 \norm{p_{in}^{\perp \tilde U}} 
\norm{t^{\perp \tilde U}} c_p   \geq - 0.008 \eps^2 \\
  & 2 \norm{p_{in}^{\perp \tilde U}} 
\norm{t^{\perp \tilde U}} c_p  \leq  \norm{p_{in}^{\perp \tilde U }}^2 + 0.008 \eps^2 \\
& c_p  \leq \frac{1}{2 \norm{t^{\perp \tilde U}} } \left( \norm{p_{in}^{\perp \tilde U }}  + 0.008 \left(\tfrac{\eps^2}{\norm{p_{in}^{\perp \tilde U }}}  \right)  \right).
\end{align*}
 
Noting that the minimum of $x + \frac{\alpha}{x}$ is $\sqrt{\alpha}$
for any fixed $\alpha$ and that $\norm{ t^{\perp \tilde{U}}}$ is $O(
\sqrt{d} \sigma)$ we obtain that $c_p \leq C \frac{\eps}{\sqrt{d}
  \sigma}$ is a sufficient constraint for suitable choice of constant
$C$.
\end{proof}

The final component is to prove that $c_p$ is indeed small. Since the
vector $t$ is random, this is generally not an issue: $t^{\tilde U}$
and $(p_{in})^{\tilde U}$ will be independent for all points $\tilde
p\in \tilde P_j\setminus \tilde P_j^s$ (but not for the sampled points
-- this is the reason they are never included in $M_j$). However, in
subsequent iterations, this may introduce some conditioning on $t_i$,
so we need to be careful and argue about all iteration ``at once''. In
fact, we show a result slightly stronger than that required by Lemma
\ref{lem:suff} (within the parameters of our model):

\begin{lemma}
Fix a ``horizon'' $j$, and a point $\tilde p=p+t$ that has not been
sampled into any set $\tilde P_l^s$ for all iterations $l\le j$
($\tilde p$ may or may not have been captured at some iteration $l\le
j$). Then $c_p\le O\left(\tfrac{\sqrt{\log n}}{\sqrt{d}}\right)$ at
iteration $j$, with high probability.
\end{lemma}
\begin{proof}
The assumption on not having been sampled means that we can think of
running the algorithm, disallowing $\tilde p$ to be in the sample. In
this situation, we can run the entire algorithm, up to iteration $j$,
without $\tilde p$ --- it does not affect the rest of the points in
any way. Hence, the algorithm can sample sets $\tilde P_0^s,\ldots
\tilde P_j^s$, and fix spaces $\tilde U_0^s,\ldots \tilde U_j^s$,
without access to $\tilde p$. Furthermore, for each $l=0\ldots j$, the
vector $p_{in}^{\perp\tilde U_l^s}$ is some vector, {\em independent}
of the noise $t$. Hence we can ``reveal'' the vector $t$, after having
fixed all vectors $p_{in}^{\perp\tilde U_l^s}$, for $l=0\ldots j$. The
vector $t$ will have angle $O(\tfrac{\sqrt{\log n}}{\sqrt{d}})$ with
all of them, with high probability. Note that, at this moment, it does
not actually matter whether $\tilde p$ was captured early on or
not. (Again, $t$ is admittedly conditioned via bounds on $\|T\|$ and
$\|t\|$, but since these are ``whp'' events, they do not affect the
conclusion.)
\end{proof}

\subsubsection{Algorithm performance}

We now remark on the resulting parameters of the algorithm. 

Processing an iteration of the preprocessing stage takes
$O(rd^2+d^3+ndk)=O(nd^2)$ time for: computing $\tilde P_j^s$, the PCA
space, and $M_j$ respectively. Hence, over $O(\sqrt{d\log n})$
iterations, together with preprocessing of the $k$-dimensional NNS
data structures, we get preprocessing time $O((nd^2+d^3+\Fprep)\sqrt{d
  \log n})$.

Space requirement is essentially that of $O(\sqrt{d\log n})$ instances
of $k$-dimensional NNS data structure, plus the space to store
$O(\sqrt{d\log n})$ spaces $\tilde U_j^s$, and the left-over set $R$.

The query time is composed of: computing the projections into
$O(\sqrt{d\log n})$ subspaces, querying the $k$-dimensional NNS data
structures, and computing the distances to left-over points in
$R$. Overall, this comes out to $O(dk\cdot \sqrt{d\log n}+\sqrt{d\log
  n}\cdot F_{\text{query}}+d|R|)$.

\section{PCA tree}\label{sec:pcaTree}

We now present our second spectral algorithm, which is closely related
to the PCA tree \cite{Sproull-pcaTree, verma2009spatial}. We first
give the algorithm and then present its analysis. Overall, we prove
the following theorem.

\begin{theorem}\label{thm:pcaTree}
Consider the Gaussian-noise model \eqref{eq:tildeP}-\eqref{eq:pstar},
and assume its parameters satisfy $\sigma<\kappa\cdot \min\left\{
\tfrac{\eps}{\sqrt{k\log n}},\ \tfrac{ \eps}{\sqrt{k}\sqrt[4]{d\log
    n}}\right\}$, for sufficiently small constant $\kappa>0$.  There
exists a data structure that preprocesses $\tilde P$, and then given
the query $\tilde q$, returns the nearest neighbor $\tilde p^*$
\whp\footnote{The probability is over the randomness from the model.}
And \whp the query time is $(k/\eps)^{O(k)}\cdot d^2$, the space
requirement is $O(nd)$, and the preprocessing time is $O(n^2d+nd^3)$.
\end{theorem}

The algorithm itself is deterministic.

\subsection{Algorithm description}
The algorithm constructs one-space partitioning tree hierarchically,
where each tree node is associated with a subset of the pointset
$\tilde P$.  We start with the root of the tree, associated with all
$n$ points $\tilde P$.  Now at each tree node $x$, we take the
pointset associated with $x$, termed $\tilde P^{in}_x$.  First, we
perform a process called ``de-clumping'', which just discards part of
the dataset, to obtain a set $\tilde P_x\subseteq \tilde P^{in}_x$. We
describe this process at the end.

The main operation at a node is to take the top centered-PCA direction
of $\tilde P_x$, termed $v_x$. By centered-PCA we mean 
subtracting from each vector in $\tilde P_x$ their average $a=\tfrac{1}{|\tilde P_x|}\sum_{\tilde p\in \tilde P_x} \tilde p$, and then taking
the top PCA direction. Now, let
$\theta\triangleq\tfrac{\eps}{1000k^{3/2}}$ and let $\Theta$ be the
partition of the real line into segments of length $\theta$, namely
$\Theta=\{[\theta i, \theta (i+1))\mid i\in \Z\}$. Then we partition
  $\tilde P_x$ into parts depending on which segment from $\Theta$ the
  projection of a point $\tilde p\in \tilde P_x$ onto $v_x$ falls
  into. Then, we orthogonalize with respect to $v_x$, namely,
  transform each point $\tilde p\in \tilde P_x$ into $\tilde p'=\tilde
  p-\langle\tilde p, v_x\rangle v_x$. 
For each non-empty segment of $\Theta$ we produce a child of $x$ associated with the points that fall into that segment, 
and repeat recursively on it. 
We stop once the current tree node has at most $d$ points associated with it.

During a query, we follow the tree into all the buckets (slabs) that
intersect a ball of radius $1+\eps/2$ around $\tilde q$. In each
leaf, compute the exact distance from $q$ to all points associated to 
that leaf. Finally, report the closest point found.

We now describe the de-clumping procedure that is done at each
node. We compute the top centered-singular value of $\tilde
P^{in}_x$. If this value is at least $\lambda_c=\lambda_c(|\tilde
P^{in}_x|)\triangleq\tfrac{\eps}{16}\sqrt{|\tilde P^{in}_x|/k}$, then
set $\tilde P_x\triangleq \tilde P_x^{in}$. Otherwise, find the closest
pair of points in $\tilde P_x^{in}$, and let $\delta$ denote their squared-distance. Remove all the pairs of points in $\tilde P_x^{in}$ that
have squared-distance at most $\delta+\eps^2/2$, to obtain $\tilde
P_x$. (The removal is iterative, proceeding in arbitrary order.)

\subsection{Analysis: Tree Depth}

The key to the analysis is to show that our PCA tree has depth at most $2k$. 
The rest of analysis will follow as we show in later sections.

In the analysis, we use the centered-PCA directions. For this purpose,
we first define the centering operation $c(A)$ for a set/matrix of
points: $c(A)\triangleq A-\tfrac{1}{|A|}\sum_{p\in A}p$. Then the
centered singular value, denoted $\|A\|_c$, is $\|c(A)\|$. Note that
the norm still satisfies the triangle inequality.

\begin{lemma}[Tree Depth]\label{lem:depth}
The constructed PCA tree has depth at most $2k$.
\end{lemma}

\begin{proof}
We first analyze the effect of orthogonalization on the points $\tilde
p$. Fix some node $x$ at a level $1\le l\le 2k$, and some point
$\tilde p=p+t_p$ reaching it.
Call $\tilde p^x\in \tilde P_x^{in}$ as its version at node $x$, after
the anterior orthogonalizations at the ancestor nodes. Also, define
$n_x\triangleq |\tilde P_x|$. 

We view each step of orthogonalization as two displacement
processes. If we have orthogonalized the point with respect to some
vector $v$, this is equivalent to snapping (projecting) the point
$\tilde p^x$ to the hyperplane defined by $\{z\in \R^d\mid
zv=\theta\cdot \lfloor \tfrac{\tilde p^x\cdot v}{\theta}\rfloor\}$, and
then moving the hyperplane towards the origin. Most importantly, all
points from node $x$ going to the same child will be all snapped to
the same hyperplane. The snapping to the hyperplane moves the point
$\tilde p^x$ by a vector of norm at most $\theta$. Note that, while the
algorithm also moves the hyperplane to pass through the origin, this
does not change the relative distances of the points in this
hyperplane.

Thus we can write each point $\tilde p^x$ reaching node $x$ as $\tilde
p^x=\tilde p+m^x+m_p^x$, where $m^x$ is the sum of all the hyperplane
moves (and is dependent on the node $x$ only), and $m^x_p$ which is
the sum of the ``snapping moves'' and depends on the actual point. We
observe that $m^x_p$ has small norm, and, in particular
$\|m_p^x\|^2\le l\cdot \theta^2\le 2k\theta^2$, since each move is in
an orthogonal direction with respect to the previous moves.

Below we assume that Lemma \ref{lem:subsetNorm} holds. Also, for any
two points $p_1,p_2$, the norm of difference of the noises is
$\|t_{p_1}-t_{p_2}\|^2=2\sigma^2d\pm 0.1\eps^2$ according to Corollary
\ref{cor:projec} for $\sigma\ll \eps/\sqrt[4]{d\log n}$.

The main part of the proof is to prove that the top PCA direction
$v_x$ at each node $x$ is close to $U$. We prove this by induction
over levels.
\begin{claim}[Induction hypothesis]
\label{clm:projection}
Consider some node $x$ at level $l$, which contains at least
$d=\Omega(\log n)$ points. Let $v_x$ be the top centered-PCA direction of
$\tilde P_x$.  The projection of $v_x$ onto $U^\perp$ is at most
$\gamma=O(\sigma\sqrt{\log n}\cdot \sqrt{k}/\eps)$.
\end{claim}

Before proving the induction hypothesis, we need to show an additional
claim, which characterizes the result of de-clumping in the current
node $x$: essentially that the top PCA direction is heavy in the
space $U$. For the claim below, we assume that Claim
\ref{clm:projection} is true at all levels {\em above} the current
node $x$.

For a node $x$, we define helper sets/matrices $P_x, M_x, T_x$ as
follows. First, consider points $\tilde P_x$, take their non-noised
versions living in $U$ (as in the model), and move using the vector
$m^x$; this gives the set $P_x$. Also, let $T_x$ be the noise vectors
of $\tilde P_x$. Define matrix $M_x$ as being composed of movements
$m_p^x$ for all points $p$ in $\tilde P_x$. Note that $\tilde
P_x=P_x+M_x+T_x$, and that $\|T_x\|_c\le\|T_x\|\le \eta(n_x)\le
\lambda_c(n_x)$, where $\eta(n_x)$ is the function from Lemma
\ref{lem:subsetNorm}.

\begin{claim}
\label{clm:lambda}
Suppose we performed the de-clumping step on $\tilde P_x^{in}$, to obtain
$\tilde P_x$.  For $v_x$ the top centered-PCA direction of $\tilde P_x$, we
have that $\|c(P_x) v_x\|\ge \lambda_c-\eta(n_x)$.
\end{claim}
\begin{proof}
Suppose the top singular value of $c(\tilde P_x^{in})$ is at least
$\lambda_c$ (in which case no de-clumping is done and $\tilde
P_x=\tilde P^{in}_x$). 
Hence, considering $P_x=\tilde P_x-(\tilde P_x-P_x)$, which also implies $c(P_x)=c(\tilde P_x)-c(\tilde P_x-P_x)$, we have 
$$\|c(P_x)v_x\|\ge \|c(\tilde P_x)v_x\|-\|c(\tilde P_x-P_x)v_x\|\ge
\lambda_c-\eta(|\tilde P_x|),$$ since $\|c(\tilde
P_x-P_x)v_x\|=\|c(T_x)v_x\|\le \eta(|\tilde P_x|)$ by Lemma
\ref{lem:subsetNorm}, and the fact that $M_xv_x=0$.

Otherwise, the algorithm throws out some points from $\tilde
P_x^{in}$.  Define $P^{in}_x$ similarly to $P_x$: take original
(non-moved, no noise) versions of the points from $\tilde P_x^{in}$,
plus the overall movement $m^x$.  In this situation, there must be two
points $p_1,p_2\in P^{in}_x$ such that $\| p_1-p_2\|\le \eps/4$:
otherwise, the top singular value of $\tilde
P_x^{in}=P_x^{in}-(P_x^{in}-\tilde P_x^{in})$ would be at least
$\|P_x^{in}\|_c-\eta(|\tilde P_x^{in}|)\ge \tfrac{\|p_1-p_2\|}{2}\cdot
\sqrt{|\tilde P^{in}_x|/k}-\eta(|\tilde P_x^{in}|)\ge \lambda_c$, a
contradiction.

We now want to characterize the minimal distance $\delta$ in $\tilde
P_x^{in}$. Note that the projection of $m_{p_1}^x$ and $m_{p_2}^x$
into $U^\perp$ is at most $2k\theta\gamma$, since each of the basis
vectors of $m_p^x$ has projection into $U^\perp$ at most
$\gamma$. Hence, the square of the component of $\tilde p_1-\tilde
p_2$ in $U^\perp$ is equal to:
$$
(\|t_{p_1}-t_{p_2}\|\pm 2k\theta\gamma)^2
=
2\sigma^2d\pm 0.1\eps^2\pm (2k\theta\gamma)^2\pm 5\sigma\sqrt{d}\cdot 2k\theta\gamma
=
2\sigma^2d\pm 0.1\eps^2\pm 0.01\eps^2\pm 0.01\eps^2,
$$
for $\sigma\ll \eps/\sqrt{\log n}$ and $\sigma\ll\tfrac{\eps}{\sqrt[4]{d\log n}}$.
Thus, for $\tilde p_1=p_1+t_1$ and $\tilde p_2=p_2+t_2$:
\begin{equation}\label{eqn:deltaLB}
\delta
=
\|\tilde p_1-\tilde p_2\|^2
\ge
(\|t_1-t_2\|- 2k\theta\gamma)^2
\ge 
2\sigma^2d-0.12\eps^2.
\end{equation}

After the de-clumping, the distance between any two distinct points
$p',p''\in P_x$, with noise vectors $t',t''$ respectively, must satisfy:
$$
(\|p'-p''\|+ 2\sqrt{2k}\theta)^2\ge \delta+\eps^2/2 - (\|t'-t''\|+ 2k\theta\gamma)^2
\ge
\eps^2/2-2\cdot 0.12\eps^2 \ge \eps^2/4.
$$

Hence $\|p'-p''\|\ge \eps/2-0.01\eps>\eps/4$, which means that
$\|\tilde P_x\|_c\ge \lambda_c$ (as already argued above). Hence we
can apply the same argument as above, this time for $\tilde P_x$
instead of $\tilde P^{in}_x$.
\end{proof}

We now prove the induction hypothesis, namely Claim
\ref{clm:projection}, for the current node $x$.

\begin{proof}[Proof of Claim \ref{clm:projection}]
Let $\tilde P_x$ be the points contained in $x$.  By Claim \ref{clm:lambda}, we have
$\lambda_{PM}\triangleq \|P_x+M_x\|_c\ge
\|c(P_x+M_x)v_x\|=\|c(P_x)v_x\|\ge\lambda_c-\eta(|\tilde P_x|)\ge \lambda_c/2$. 

Decompose the top centered-PCA direction $v_x$ of $\tilde P_x$ as $v_x=\sqrt{1-\alpha^2}
u+\alpha u'$, where $u\in U$ and $u'\perp U$ are of norm
1. Note that $\alpha$ is exactly the projection of $v_x$ onto $U^\perp$.

We will bound $\alpha$ by lower and upper bounding $\|\tilde P_x\|_c$
as a function of $\alpha$ and $\lambda_{PM}$.  We start with the upper
bound on $\|\tilde P_x\|_c$. For this we decompose the matrix
$c(P_x+M_x)$ into the component in $U$, called $c(P_x+M_x^U)$, and the
perpendicular one, called $c(M_x^\perp)$. Note that $c(P_x)$ lies
inside $U$, despite the movement $m_x$, because of the centering
$c(\cdot)$. We now bound the spectral norm of $\|c(M_x^\perp)\|$,
using the inductive hypothesis that the projection of each of at most
$2k$ basis vectors of $m_p^x$ onto $U^\perp$ is at most $\gamma$:
$$
\|c(M_x^\perp)\|\le \max_p
\sqrt{n_x}\|m_p^x\|\cdot \sqrt{2k}\gamma
\le
\gamma\cdot \sqrt{n_x}2k\cdot \theta
\le
\gamma\cdot \tfrac{\eps}{500}\sqrt{n_x/k}
\le
\tfrac{\gamma}{9}\cdot \tfrac{\eps}{32}\sqrt{n_x/k}
\le
\tfrac{\gamma}{9}\lambda_{PM}.
$$
Thus, we have that
$\lambda_{PM^U}\triangleq\|P_x+M_x^U\|_c=
\lambda_{PM}\pm\tfrac{\gamma}{9} \lambda_{PM}$.

We can now upper bound the norm of $\tilde P_x$:
\begin{equation}\label{eqn:upperPxStart}
\|\tilde P_x\|_c
\le
\|c(P_x+M_x)v_x\| + \|c(T_x)v_x\|
\le
\|c(P_x+M_x)v_x\| + \alpha\eta(n_x).
\end{equation}

We need to bound the first term now. For this, we compute the
following ratio, using that the projection of $v_x$ into $U^\perp$ is of magnitude 
$\alpha$:
\begin{eqnarray}\label{eqn:upperPxFinal}
\nonumber
\frac{\|c(P_x+M_x)v_x\|^2}{\lambda_{PM^U}^2}
&\le
\frac{\|c(P_x+M_x^U)v_x\|^2+\|c(M_x^\perp)v_x\|^2+2\|c(P_x+M_x^U)v_x\|\|c(M_x^\perp)v_x\|}{\lambda_{PM^U}^2}
\\
\nonumber
&\le
\frac{(1-\alpha^2)\lambda_{PM^U}^2+\alpha^2\cdot (\gamma/9)^2\cdot \lambda_{PM}^2+2\alpha\sqrt{1-\alpha^2}\cdot \lambda_{PM^U}\cdot \tfrac{\gamma}{9}\lambda_{PM}}{\lambda_{PM^U}^2}.
\\
\nonumber
&\le
(1-\alpha^2)+\alpha^2(\gamma/9)^2(1+\gamma/9)^2+2\alpha\tfrac{\gamma}{9}\cdot (1+\gamma/9)
\\
&\le
1-\alpha^2/2+\alpha\gamma/3.
\end{eqnarray}


On the other hand, we want to prove a lower bound on $\|\tilde
P_x\|_c$. We define $u_{PM^U}$ to be the centered-PCA direction of $P_x+M_x^U$:
$\lambda_{PM^U}=\|(P_x+M_x^U)u_{PM^U}\|_c$. Remember that $u_{PM^U}$
lies inside $U$.

\begin{equation}\label{eqn:lowerPx}
\|\tilde P_x\|_c
\ge 
\|c(\tilde P_x) u_{PM^U}\|
=
\|c(P_x+M_x^U) u_{PM^U}\|
=
\lambda_{PM^U}.
\end{equation}

Putting together the upper bound \eqref{eqn:upperPxStart}, \eqref{eqn:upperPxFinal} and the lower bound \eqref{eqn:lowerPx}, we obtain 
$$
\lambda_{PM^U}\le \lambda_{PM^U}\sqrt{1-\alpha^2/2+\alpha\gamma/3}+\alpha\eta(n_x)
$$
and, using that $\sqrt{1-x}\le 1-x/2$ for all $0 \leq x \leq 1$, we conclude:
$$
\alpha\le 4\cdot\left(\tfrac{\eta(n_x)}{\lambda_{PM^U}}+\gamma/6\right)\le \gamma,
$$ 
as long as $4\tfrac{\eta(n_x)}{\lambda_{PM^U}}<\gamma/3$. Since
$\lambda_{PM^U}\ge \lambda_{PM}/2\ge \lambda_c/4$, this is
satisfied if
$$
\gamma=\Theta(\sigma \sqrt{\log n}\sqrt{k}/\eps).
$$
We are done proving the inductive hypothesis.
\end{proof}

The final step is to show that, because all vectors $v_x$ along a
root-to-leaf path are perpedicular to each other and are heavy in $U$,
there cannot too many of them, and hence the tree depth is bounded.

\begin{claim}
Fix some dimension $k>1$ subspace $U\subset \R^d$. Suppose there exists
$2k+1$ vectors $v_1,\ldots v_{2k+1}$, such that the projection of each
into $U$ has norm at least $1-1/4k$. Then at least two of the vectors are
not orthogonal to each other.
\end{claim}
\begin{proof}
For the sake of contradiction assume that all vectors $v_1\ldots
v_{2k+1}$ are mutually orthogonal. Then let $u_i$ be the projection of
$v_i$ into $U$, and let $\delta_i=v_i-u_i$. We want to bound the dot product
$u_iu_j$. Consider
$$
0=v_i^\tran\cdot v_j=(u_i^\tran+\delta_i^\tran)(u_j+\delta_j)=u_i^\tran u_j+\delta_i^\tran\delta_j.
$$ Since $|\delta_i^\tran \delta_i|\le 0.25/k$, we have that $|u_i^\tran u_j|\le
0.25/k$. Even after normalization, we have that
$\left|\tfrac{u_i^\tran u_j}{\|u_i\|\cdot \|u_j\|}\right|\le 1/k$. Following Alon's
result \cite{Alon}, we conclude that $u_i$'s have rank at least
$(2k+1)/2>k$, which is impossible if all $u_i$ live in the
$k$-dimensional space $U$. Thus we have reached a contradiction.
\end{proof}

We now combine the two claims together. Note that $\alpha\le 1/4k$. We
conclude that the height of the tree is at most $2k$, thus concluding
the proof of Lemma \ref{lem:depth}.
\end{proof}




\subsection{Analysis: Correctness}

Now that we have established an upper bound on the tree depth, we will
show that the query algorithm will indeed return the nearest neighbor
$\tilde p^*$ for the query $\tilde q$ (modelled as in Section
\ref{sec:model}). We show this in two steps: first we prove the
result, assuming the point $\tilde p^*$ was not thrown out during
de-clumping. Then we show that the de-clumping indeed does not throw
out the point $\tilde p^*$.

\begin{lemma}
The query algorithm returns the point $\tilde p^*$, assuming it was
not thrown out during the de-clumping process.
\end{lemma}
\begin{proof}
Consider some non-leaf tree node $x$ containing $\tilde p^*$. We need
to argue that, at node $x$, the query algorithm follows the child
where $\tilde p^*$ goes.

As before, let $\tilde p^x$ be the orthogonalized version of $\tilde
p^*$ (above $x$) and $m^x_p$ is the total amount of ``hyperplane
snapping'' that happened to $\tilde p^*$. 
We also have that $v_x$ has projection at most
$\gamma$ onto $U^\perp$ (from Claim \ref{clm:projection}). Hence, we have:
$$ |v_x\tilde p^x-v_x\tilde q|\le
|v_x(t_{p^*}-t_{q})|+|v_x(p^*-q)|,$$
again using that $v_xm_p^x=0$. Note that, with high probability,
$$|v_x(t_{p^*}-t_{q})|\le 3\sigma\sqrt{d}\cdot
\gamma\le \eps/8\sqrt{k}.
$$
Since this is true also for all ancestors $y$ of $x$, and since all $v_y$, together with $v_x$, are mutually orthogonal, we have that:
\begin{align*}
\sum_{y\hbox{ ancestor of }x}&|v_y(\tilde p^x-\tilde q)|^2
\\
&=\sum_{y\hbox{ ancestor of }x}|v_y(t_{p^*}-t_{q})+v_y(p^*-q)|^2
\\
&\le
\left(\sqrt{\sum_{y\hbox{ ancestor of }x} \left|v_y(m^y_p+|t_{p^*}-t_{q}|)\right|^2}+\sqrt{\sum_{y\hbox{ ancestor of }x}|v_y(p^*-q)|^2}\right)^2
\\
&\le
\left(\sqrt{2k\cdot \left(\eps/8\sqrt{k}\right)^2}+\sqrt{\sum_{y\hbox{ ancestor of }x}|v_y(p^*-q)|^2}\right)^2
\\
&\le
(\eps/4\sqrt{2}+1)^2
\\
&<
\eps/2+1.
\end{align*}
This means that the bucket
(child node of $x$) of $\tilde p^x$ intersects a ball of radius
$1+\eps/2$ around $\tilde q$, and hence the query algorithm will go
into that bucket (child).
\end{proof}

To complete the correctness argument, we also need to prove that $p^*$
is never thrown out due to the de-clumping process.

\begin{claim}
$p^*$ is never thrown out due to the de-clumping process.
\end{claim}
\begin{proof}
For the sake of contradiction, suppose $\tilde p^*_x$ is thrown out at
some node $x$. This means there is some point $\tilde p_x$ such that
$\|\tilde p_x-\tilde p^*_x\|^2\le \delta+\eps^2/2$. Since
$\|q-p\|-\|q-p^*\|\ge \eps$, we have that $\|p-p^*\|\ge \eps$.

We have:
\begin{align*}
\delta+\eps^2/2\ge \|\tilde p_x-\tilde p^*_x\|^2
&= 
\|p-p^*+(m_p^x-m_{p^*}^x)+(t_p-t_{p^*})\|^2
\\
&=
\|(p-p^*+(m_p^x-m_{p^*}^x))U\|^2+\|(m_p^x-m_{p^*}^x+t_p-t_{p^*})U^\perp\|^2
\end{align*}

We want to bound $\|(m_p^x-m_{p^*}^x+t_p-t_{p^*})U_\perp\|^2$. Indeed,
note that the projection of $m_p^x$ onto $U^\perp$ can be at most
$2k\theta\cdot \gamma\le O(\sigma \sqrt{\log n})$, by Claim
\ref{clm:projection}. Hence:
\begin{align*}
\|(m_p^x-m_{p^*}^x+t_p-t_{p^*})U_\perp\|
&\ge 
-O(\sigma \sqrt{\log n})
+\sqrt{2\sigma^2d- 0.1\eps^2}
\\
&\ge
-O(\sigma \sqrt{\log n})
+\sigma\sqrt{2d}- \tfrac{0.06\eps^2}{\sigma\sqrt{2d}}
\\
&\ge
\sigma\sqrt{2d}- \tfrac{0.09\eps^2}{\sigma\sqrt{2d}},
\end{align*}

as long as $O(\sigma \sqrt{\log n})<
\tfrac{0.03\eps^2}{\sigma\sqrt{2d}}$, which is equavalent to saying
that $\sigma\le O\left(\tfrac{\eps}{\sqrt[4]{d}\sqrt[4]{\log n}}\right)$.

Putting it all together, we have:
\begin{align*}
\delta+\eps^2/2\ge \|\tilde p_x-\tilde p^*_x\|^2
&\ge (\eps-2\cdot \sqrt{2k}\theta)^2+(\sigma\sqrt{2d}- \tfrac{0.09\eps^2}{\sigma\sqrt{2d}})^2
\\
&\ge
0.8\eps^2+2\sigma^2d-0.2\eps^2
\\
&> 2\sigma^2d+0.6\eps^2.
\end{align*}

But, as proven in Eqn. \eqref{eqn:deltaLB}, we have that $\delta\le
2\sigma^2d+0.12\eps^2$. We've reached a contradiction.
\end{proof}

\subsection{Analysis: Performance}

The space and preprocessing bounds follow immediately from the
construction. We just need to argue about the query time.

\begin{claim}
The query time is $(k/\eps)^{O(k)} d^2$.
\end{claim}
\begin{proof}
At each node of the tree, there are at most $O(1/\theta)=O(k^{3/2}/\eps)$
child nodes that are followed. Hence, in total, we reach
$O(1/\theta)^{2k}=(k/\eps)^{O(k)}$ leaves. The factor of $d^2$ comes
from the fact that each leaf has at most $d$ points to check the
distance against.
\end{proof}





{\small
\bibliographystyle{alphaurlinit}
\bibliography{nn,andoni}
}

\newpage
\appendix

\section{Additional Material for Section \ref{sec:bounded}}
\label{app:bounded}

\begin{proof} [Proof of Claim \ref{cl:boundedReport}]
Fix $p\neq p^*$ that is captured by the same $\tilde U$,
and use the triangle inequality to write 
$\norm{p - \tilde{p}_{\tilde{U}}} 
  \leq \norm{p-\tilde{p}} + \norm{ \tilde{p} - \tilde{p}_{\tilde{U}} }
  \leq \alpha + \sqrt{2} \alpha
  \leq 3\alpha
$.
Doing similarly for $p^*$ we get
$\norm{p^* - \tilde p^*_{\tilde{U}}} \leq 3\alpha$,
and by our assumption
$\norm{q - \tilde q} \leq \alpha$.
Using all of these and the triangle inequality once again, 
we bound (in some sense ``comparing'' $p$ vs.\ $p^*$)
\begin{align*} \label{eq:bounded1}
  \frac{ \norm{ \tilde q - \tilde p_{\tilde U} } }
       { \norm{ \tilde q - \tilde p^*_{\tilde U} } }
  &= \frac{ \norm{ q-p } \pm 4\alpha  }
         { \norm{ q-p^* } \pm 4\alpha } 
  =    \frac{ \norm{ q-p } \pm \tfrac14 \eps  }
            { \norm{ q-p^* } \pm \tfrac14 \eps } \\
& \geq \frac{\norm{q-p} - \tfrac14 \eps}
	{\norm{q-p^*} + \tfrac14 \eps}
  \geq \frac{ \norm{ q-p^* } + \tfrac34 \eps }
            { \norm{ q-p^* } + \tfrac14 \eps }
  > 1+\tfrac14 \eps.
\end{align*}
Aiming to ``replace'' $\tilde q$ above with $\tilde{q}_{\tilde{U}}$,
its projection of $\tilde q$ onto $\tilde U$, 
we use Pythagoras' Theorem 
(recall both $\tilde p_{\tilde U},\tilde p^*_{\tilde U}\in \tilde U$),
\begin{align*}
  \frac{ \norm{ \tilde{q}_{\tilde{U}} - \tilde{p}_{\tilde{U}} }^2 }
       { \norm{ \tilde{q}_{\tilde{U}} - \tilde p^*_{\tilde{U}} }^2 }
  = \frac{ \norm{ \tilde q - \tilde p_{\tilde U} }^2 - \norm{ \tilde q - \tilde q_{\tilde U} }^2  }
          { \norm{ \tilde q - \tilde p^*_{\tilde U} }^2 - \norm{ \tilde q - \tilde q_{\tilde U} }^2  } 
  > (1+\tfrac14 \eps)^2,
\end{align*}
where the last inequality crucially uses the fact that 
the numerator and the denominator contain the exact same term 
$\norm{ \tilde q - \tilde q_{\tilde U} }^2$.
We conclude that $\tilde{p}^*$ is indeed reported by 
the $k$-dimensional data structure it is assigned to.
\end{proof}

\section{Short Review of Spectral Properties of Matrices}
\label{apx:spectral}

We review some basic definitions of spectral properties of matrices.

\subsection{Spectral Norm and Principal Component Analysis}

The \emph{spectral norm} of a matrix $X\in\R^{n\times d}$ is defined
as $\norm{X} = \sup_{y\in \R^d: \norm{y}=1} \norm{X y}$, where all
vector norms $\norm{\cdot}$ refer throughout to the $\ell_2$-norm.
The Frobenius norm of $X$ is defined as $\norm{X}_F=(\sum_{ij}
X_{ij}^2)^{1/2}$, and let $X^\tran$ denote the transpose of $X$.
A \emph{singular vector} of $X$ is a unit vector $v \in \reals^d$ 
associated with a \emph{singular value} $s\in\reals$ and a unit vector 
$u \in \reals^n$ such that $X v = s u$ and $u^\tran X = s v^\tran$. 
(We may also refer to $v$ and $u$ as a pair of right-singular and left-singular vectors associated with $s$.)

\begin{fact}\label{tria}
For every two real matrices $X$ and $Y$ of compatible dimensions
(i) $\|X+Y \| \leq \|X \| + \|Y \|$;\, 
(ii) $\|X  Y \| \leq \norm{X}\cdot \norm{Y}$;\, and 
(iii) $\|X \|= \| X^\tran \|$. 
\end{fact}

We can think of $i$-th row in $X$ as a point $x_i\in\reals^d$, and
define the corresponding point set $P(X) = \aset{x_1,\ldots,x_n}$.  Then a unit vector $y\in\R^d$ maximizing $\|X y\|$
corresponds to a best-fit line for the point set $P(X)$. The
PCA (and \emph{Singular Value Decomposition (SVD)} more generally) is a classical
generalization of this notion to a best-fit $k$-dimensional subspace
for $P(X)$, as described next (see e.g.\ \cite[Chapter
  1]{kannanbook}).

\begin{theorem}[\cite{kannanbook}]
\label{opt}
Let $X\in\R^{n\times d}$, and define the vectors $v_1,\ldots,v_d\in \R^d$ 
inductively by 
$$v_j = \argmax_{\norm{v} = 1;\ \forall i <j, v^\tran v_i=0} \norm{ Xv }$$
(where ties for any $\argmax$ are broken arbitrarily).
Then each $V_j=\spn\aset{v_1,\ldots,v_k}$ attains the minimum of
$\sum_{i=1}^n d(x_i,W)^2$ over all $j$-dimensional subspaces $W$.
Furthermore, $v_1,\ldots,v_d$ are all singular vectors with 
corresponding singular values $s_1=\norm{Xv_1},\ldots,s_d=\norm{Xv_d}$,
and clearly $s_1  \geq \ldots \geq s_d$.
\end{theorem}

We will later need the following basic facts (see, e.g.,
\cite{stewart}).  We denote the singular values of a matrix $X \in
\R^{n \times d}$ by $s_1(X) \geq s_2(X)\geq \ldots \geq s_d(X)$.

\begin{fact}\label{fa:sumsingular}
Let $P(X)\subset \reals^d$ be the point set corresponding to the rows of 
a matrix $X \in \R^{n \times d}$. Then 
\begin{enumerate} \compactify
\renewcommand{\theenumi}{(\alph{enumi})}
\item 
$\norm{X }_F^2 = \sum_{p \in P(X)}  \|p\|^2 = \sum_{i=1}^d s_i(X)^2$ 
\, and $\norm{ X } = s_1(X)$. 
\item 
$P(X)$ lies in a subspace of dimension $k$ 
if and only if $s_{k+1}(X)=0$.
\end{enumerate}
\end{fact}

\begin{fact}\label{fa:co}
For any matrix $X$, let $X^\tran X$ be the \emph{covariance} matrix of $X$. Then the right singular vectors of $X$ are also right singular vectors of $X^{\tran} X$. Also, $s_i(X^\tran X) = s_i^2(X)$.
\end{fact}

\begin{fact} \label{fa:difference}
For two matrices $X$ and $E$ of compatible dimensions, 
$\abs{s_j(X+E) - s_j(X)} \leq \norm{E}$.
\end{fact}

\subsection{Spectral Norms of Random Matrices}

In our analysis, we will also need bounds on norms of random matrices,
which we derive using standard random matrix theory.  We state below a
bound on the spectral norm of $T\in\reals^{n\times d}$, a matrix of
iid Gaussian vectors.  We also consider its restriction $T_A$ to any
subset of rows of $T$, and bound the spectral norm in terms of the
size of the subset, $s=\card{A}$, expressed as a function
$\eta(s)=O(\sigma\sqrt{s\cdot \log n})$.

\begin{theorem}[\cite{nonuniform,RV10}]
\label{randomtheory}
Let matrix $T \in \R^{n \times d}$ have entries drawn independently from $N(0, \sigma)$. Then with probability approaching $1$ asymptotically as $n$ and $d$ increase, 
$\norm{T} \leq 3 \sigma \max\aset{\sqrt{ n}, \sqrt{d}}$.
\end{theorem} 

\begin{lemma}
\label{lem:subsetNorm}
With high probability, for every subset $A$ of the rows of $T$, with $|A|\ge d$,
the corresponding submatrix $T_A$ of $T$, 
has spectral norm
$\|T_A\|\le \eta(|A|)=O(\sigma\sqrt{|A|\cdot \log n})$.
\end{lemma}
\begin{proof}
Fix some $A$ of size $s$. It is known from random matrix theory that 
$\|T_A\|\le \sigma\sqrt{Cs\log n}$ with probability at least
$1-e^{-\Omega(Cs\log n)}=1-n^{-\Omega(Cs)}$ \cite[Proposition 2.4]{RV10}. 
For a large enough constant $C>1$, since there are at most ${n\choose s}\le
(n/s)^{O(s)}$ such subsets $A$, by a union bound, 
with high probability none of the sets will fail. 
Another union bound over all sizes $s$ completes the claim.
\end{proof}

\end{document}